\documentclass[aps, pra, superscriptaddress, reprint, onecolumn, notitlepage]{revtex4-1}
\usepackage[cm]{fullpage}
\usepackage[T1]{fontenc}
\usepackage{amssymb, amsmath, amsthm}
\makeatletter
\def\amsbb{\use@mathgroup \M@U \symAMSb}
\makeatother
\usepackage[mathscr]{euscript}
\usepackage{mathtools}
\usepackage{xcolor}
\usepackage[colorlinks=true, linkcolor=darkred, urlcolor=darkblue, citecolor=darkgreen]{hyperref}
\usepackage{verbatim}
\usepackage{bbold}
\usepackage{enumitem}
\usepackage{framed}
\newtheorem{df}{Definition}
\newtheorem{lem}{Lemma}

\newtheorem{prop}{Proposition}

\definecolor{darkred}{RGB}{200, 0, 0}
\definecolor{darkblue}{RGB}{0, 0, 180}
\definecolor{darkgreen}{RGB}{50, 150, 0}
\definecolor{col2}{HTML}{64A857}
\definecolor{col3}{HTML}{D1603D}

\DeclareMathOperator{\tr}{tr}

\DeclareMathOperator{\id}{id}

\DeclareMathOperator{\hmin}{H_{min}}
\DeclareMathOperator{\hmineps}{H_{min}^{\varepsilon}}
\DeclareMathOperator{\Psucc}{P_{succ}^{\cF}}

\DeclareMathOperator{\pwin}{p_{win}}
\DeclareMathOperator{\pguess}{p_{guess}}
\DeclareMathOperator{\pguessseq}{p_{guess}^{seq}}
\DeclareMathOperator*{\argmax}{arg\,max}

\newcommand{\fchsh}{f_{\textnormal{CHSH}}}
\def \diracspacing {0.7pt}
\newcommand{\ave}[1]{\langle #1 \rangle}
 
\newcommand{\ket}[1]{| \hspace{\diracspacing} #1 \rangle}

\newcommand{\ketbra}[2]{| \hspace{\diracspacing} #1 \rangle \langle #2 \hspace{\diracspacing} |} 
 
\newcommand{\ketbraq}[1]{\ketbra{#1}{#1}}

\newcommand{\nbox}[2][9]{\hspace{#1pt} \mbox{#2} \hspace{#1pt}}
\newcommand{\norm}[2][]{#1| \! #1| #2 #1| \! #1|}
\newcommand{\abs}[2][]{#1| #2 #1|}

\newcommand{\post}[0]{^{\postinside}}
\newcommand{\postinside}[0]{*}
\def\width{0.6}
\def\height{0.6}
\newcommand{\textsq}[3]
{
\node[align=center] at (#1, #2) {#3};
\draw [black] (#1 + \width, #2 + \height) rectangle (#1 - \width, #2 - \height);
}

\newcommand{\cB}{\mathcal{B}}

\newcommand{\cF}{\mathcal{F}}

\newcommand{\cH}{\mathcal{H}}
\newcommand{\cI}{\mathcal{I}}

\newcommand{\cK}{\mathcal{K}}
\newcommand{\cL}{\mathcal{L}}

\newcommand{\cP}{\mathcal{P}}

\newcommand{\cS}{\mathcal{S}}

\newcommand{\cX}{\mathcal{X}}
\newcommand{\cY}{\mathcal{Y}}

\newcommand{\sH}{\mathscr{H}}

\usepackage{tikz}
\usetikzlibrary{decorations.pathmorphing}
\tikzset{>=stealth}
\newcommand{\corefig}
{
	\draw[rounded corners, darkgreen, dashed, ultra thick] (-4.2, 4.6) -- (4.7, 4.6) -- (4.7, -2) -- (-4.2, -2) -- (-4.2, 4.6);
	\node[align=center] at (0, 4.1) {\textbf{Alice}};
	\draw [->] (-2.35, 0) to (-0.95, 0);
	\textsq{-3.1}{0}{main\\device};
	\draw [->] (-3.1, 1.2) to (-3.1, 0.7);
	\node[align=center] at (-3.1, 1.45) {$\theta$};
	\draw [->] (-3.1, -0.7) to (-3.1, -1.2);
	\node[align=center] at (-3.1, -1.4) {$x$};
	\draw (0, 0) ellipse (0.8 and 0.5);
	\node[align=center] at (0, 0) {switch};
	\draw [->, in=180, out=0] (0.95, 0) to (2.85, 2.5);
	\node[align=center] at (1.65, 1.25) {$q$};
	\textsq{3.6}{2.5}{test\\device};
	\draw [->] (3.6, 3.7) to (3.6, 3.2);
	\node[align=center] at (3.6, 3.95) {$t$};
	\draw [->] (3.6, 1.8) to (3.6, 1.3);
	\node[align=center] at (3.6, 1.05) {$y$};
	\draw [->] (0.95, 0) to (6.85, 0);
	\node[align=center] at (3, -0.3) {$(1 - q)$};
	\textsq{7.6}{0}{\textbf{Bob}};
	\draw [->] (7.6, -0.7) to (7.6, -1.2);
}
\newcounter{prots}
\newenvironment{prot}[1]
{
\refstepcounter{prots}
\begin{framed}
\noindent \textbf{Protocol \arabic{prots}:}\ {\texttt{#1}}\\
}
{\end{framed}}

\begin{document}
\title{Device-independent two-party cryptography secure against sequential attacks}
\author{J\k{e}drzej Kaniewski}
\email{jkaniewski@math.ku.dk}
\affiliation{Centre for Quantum Technologies, National University of Singapore, 3 Science Drive 2, Singapore 117543}
\affiliation{QuTech, Delft University of Technology, Lorentzweg 1, 2628 CJ Delft, the Netherlands}
\affiliation{Department of Mathematical Sciences, University of Copenhagen, Universitetsparken 5, 2100 Copenhagen, Denmark}
\author{Stephanie Wehner}
\affiliation{QuTech, Delft University of Technology, Lorentzweg 1, 2628 CJ Delft, the Netherlands}
\date{\today}
\begin{abstract}
The goal of two-party cryptography is to enable two parties, Alice and Bob, to solve common tasks without the need for mutual trust. Examples of such tasks are private access to a database, and secure identification. Quantum communication enables security for all of these problems in the noisy-storage model by sending more signals than the adversary can store in a certain time frame. Here, we initiate the study of device-independent protocols for two-party cryptography in the noisy-storage model. Specifically, we present a relatively easy to implement protocol for a cryptographic building block known as weak string erasure and prove its security even if the devices used in the protocol are prepared by the dishonest party. Device-independent two-party cryptography is made challenging by the fact that Alice and Bob do not trust each other, which requires new techniques to establish security. We fully analyse the case of memoryless devices (for which sequential attacks are optimal) and the case of sequential attacks for arbitrary devices. The key ingredient of the proof, which might be of independent interest, is an explicit (and tight) relation between the violation of the Clauser-Horne-Shimony-Holt inequality observed by Alice and Bob and uncertainty generated by Alice against Bob who is forced to measure his system before finding out Alice's setting (guessing with postmeasurement information). In particular, we show that security is possible for arbitrarily small violation.
\end{abstract}
\maketitle
\section{Introduction}
\label{sec:introduction}
Quantum key distribution~\cite{bennett84,ekert91} (QKD) allows two honest parties, Alice and Bob, to protect their communication from a nosy eavesdropper. Yet, there are many other tasks that Alice and Bob may wish to solve, in which they themselves do not trust each other and secure identification is one such example. Here, Alice wants to identify herself to Bob without revealing her password. Bit commitment and oblivious transfer constitute other well-known examples of such tasks. 

It is intuitive that security for two-party cryptographic protocols is more difficult to achieve than for QKD, since Alice and Bob cannot help each other to check on the eavesdropper. Instead, every party has to fend for himself. It turns out that even using quantum communication Alice and Bob cannot achieve security without making additional assumptions~\cite{mayers97, lo97, lo97a, colbeck07a}. Usually one relies on computational assumptions, i.e.~that solving a computational puzzle requires a large amount of computing resources, namely more than is available to the adversary. Instead of relying on computational assumptions, however, it is possible to make physically motivated assumptions, for example that the adversary's ability to store information is limited. Introducing such storage restrictions was pioneered by Maurer~\cite{maurer91}, who considered imposing a restriction on the adversary's ability to store \emph{classical} bits known as the bounded-storage model. Unfortunately, the fact that (i) classical storage is cheap and plentiful and (ii) the gap between what the honest parties need to implement the protocol and what a dishonest party needs to break it is only polynomial~\cite{cachin97}, renders this model less practical. In contrast, storing quantum information reliably is an extremely difficult problem, motivating the so-called bounded-quantum storage~\cite{damgard05,damgard07} or more generally
noisy-storage model~\cite{wehner08a,konig12}. The noisy-storage model admits protocols that require no quantum storage for the honest execution and that can be implemented in a manner similar to QKD using BB84~\cite{wehner08a,konig12,dupuis15}, six-state~\cite{berta14} or continuous variable~\cite{furrer15} encodings. Significantly, security can always 
be achieved as long as the number of qubits $n$ sent in the protocol is only slightly larger than the number of qubits $r$ that the adversary can store, that is, whenever $r \lesssim n - O(\log n)$~\cite{dupuis15}, which is essentially optimal. First implementations of bit commitment~\cite{ng12a} and oblivious transfer~\cite{erven14} in the noisy-storage model have been demonstrated. Note that there exist other assumptions that make two-party cryptography possible, e.g.~that the two parties are given access to guaranteed additional resources~\cite{rivest99, crepeau97, winter03}, or that they must delegate agents who cannot communicate during the protocol (which might be motivated by special relativity)~\cite{benor88, kent99, kent05, simard07, crepeau11, kent11, kent12, kaniewski13, kaniewski15a}. The noisy-storage model is particularly interesting since in contrast to computational or relativistic assumptions, security is preserved even if the assumption is invalidated at a later point. That is, security cannot be broken retroactively if the adversary acquires a larger quantum storage device in the future, making this assumption completely future-proof.

One of the central questions in (quantum) cryptography is finding the minimal assumptions which are sufficient to guarantee security. For example in the standard QKD scenario we assume that the quantum channel between Alice and Bob is untrusted (i.e.~it is fully controlled by the eavesdropper) but the devices used by Alice and Bob inside their laboratories are fully characterised. Already early on, however, it was recognized that violation of a Bell inequality is intimately linked to cryptographic security~\cite{ekert91}. Mayers and Yao~\cite{mayers98, mayers04} went on to realise that quantum states can be \emph{self-tested}, i.e.~that certain quantum properties can be verified by a purely classical user, which started the field of \emph{device-independent} (DI) quantum cryptography. In DI cryptography instead of assuming that we know how the devices work, we simply test them during the protocol by using them to exhibit Bell nonlocality~\cite{brunner14}. DI cryptography has been one of the most active research topics within quantum cryptography, predominantly in the context of QKD~\cite{barrett05a, acin06, acin07, barrett13, reichardt13, vazirani14, miller14, miller14a, aguilar15} and randomness expansion or amplification~\cite{colbeck10, pironio10, vazirani12, coudron13, bouda14, miller14, miller14a}.

DI two-party cryptography, on the other hand, remains a largely unexplored territory. Security of a protocol for imperfect coin flipping and bit commitment has been analysed in the DI regime~\cite{silman11, aharon16}. Significantly, the setting considered by these works is different: since the authors do not impose any extra assumptions, they cannot hope to reach the perfect primitive so they aim for an imperfect implementation instead. Moreover, Adlam and Kent have recently proposed a DI relativistic bit commitment protocol~\cite{adlam15b}, which
allows security for a fixed amount of time under the assumption that each party is split into space-like separated agents.

Here, we take the very first step in proving DI security for two-party cryptographic protocols in the noisy-storage model. That is, we establish the security of these protocols even if the devices are not trusted under some extra assumptions (either we require the devices to behave identically in every round or we require the attack of the dishonest party to be sequential). To accomplish this, there are a number of conceptual as well as technical hurdles to cross.

\begin{enumerate}
\item In QKD Alice and Bob are always honest, while Eve is always trying to break the protocol. In DI QKD it is therefore natural
to give the power to prepare the devices to Eve. Analogously, we will assume here that all the devices used in the protocol are
always prepared by the dishonest party.
\item In the following section we will see that the protocol we start with uses quantum communication between Alice and Bob. This means that
the adversary who prepared the devices will receive \emph{quantum communication} coming back from the devices. This is in sharp contrast
to DI QKD, in which Eve prepares the devices -- with which she is possibly entangled -- and then Alice and Bob simply 
push buttons on the devices to perform measurements. That is, there is no quantum communication going back to Eve. This feature introduces a significant difference between the security analysis of DI QKD and DI two-party cryptography protocol considered here and requires us to develop novel proof techniques.
\end{enumerate}
\subsection{Results}
To establish DI security of two-party protocols, we will establish the DI security for a \emph{universal} two-party primitive known as weak string erasure (WSE)~\cite{konig12}. The most convenient manner of describing a new primitive is to specify its input-output behaviour. Such an abstract description is known as the \emph{ideal functionality} and the ideal functionality of WSE is explained in Fig.~\ref{fig:wse}. Universality means that a secure implementation of WSE can be used to construct \emph{any} other two-party cryptographic primitive. In particular, the well-known primitive of bit commitment can be obtained from WSE using classical post-processing. Since classical post-processing is trusted in the model of DI quantum cryptography, this means that once we construct a DI protocol for WSE, we have obtained a protocol for any primitive that can be obtained from WSE using classical post-processing. Moreover, the final security bound~\eqref{eq:security-memoryless} immediately implies the device-independent security of an oblivious transfer protocol in the bounded storage model (for details see Section~4.3 of Ref.~\cite{damgard07}).
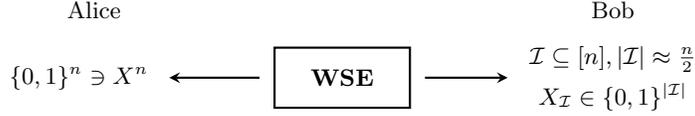
\begin{figure}[h!]
\begin{tikzpicture}[scale=1, line width=0.8]
\node at (-3.3, 0.9) {Alice};
\node at (-3.5, 0) {$\{0, 1\}^{n} \ni X^{n}$};
\draw [<-] (-2.3, 0) to (-1.1, 0);
\draw [black] (-0.9, -0.4) rectangle (0.9, 0.4);
\draw [->] (1.1, 0) to (2.2, 0);
\node at (0, 0) {\textbf{WSE}};
\node at (3.6, 0.9) {Bob};
\node at (3.6, 0.25) {$\cI \subseteq [n], \abs{\cI} \approx \frac{n}{2}$};
\node at (3.6, -0.25) {$X_{\cI} \in \{0, 1\}^{\abs{\cI}}$};
\end{tikzpicture}
\caption{The ideal functionality of WSE~\cite{konig12}: Alice gets a randomly chosen bit string $X^n$ while Bob 
obtains a randomly chosen subset of indices $\cI \subseteq [n] = \{1, 2, \ldots, n\}$ and the bits of $X^n$ corresponding to the indices in $\cI$, denoted by $X_{\cI}$. Security means that if Bob is honest, then Alice cannot learn the index set $\cI$. That is,
she does not learn which bits of the string $X^n$ are known to Bob. Conversely, if Alice is honest, then Bob finds it difficult to guess the \emph{entire} string quantified by a lower bound on the min-entropy $\hmin(X^n|\textnormal{Bob}) \geq \lambda n$ (equivalent to an upper bound on the guessing probability $\pguess(X^n|\textnormal{Bob}) \leq 2^{- \lambda n}$), where $\lambda$ is a real parameter specified by the ideal functionality. Whenever $\lambda > 0$, WSE is useful for constructing
other cryptographic primitives like bit commitment. We defer formal definitions until Section~\ref{sec:security-definitions}.
}
\label{fig:wse}
\end{figure}
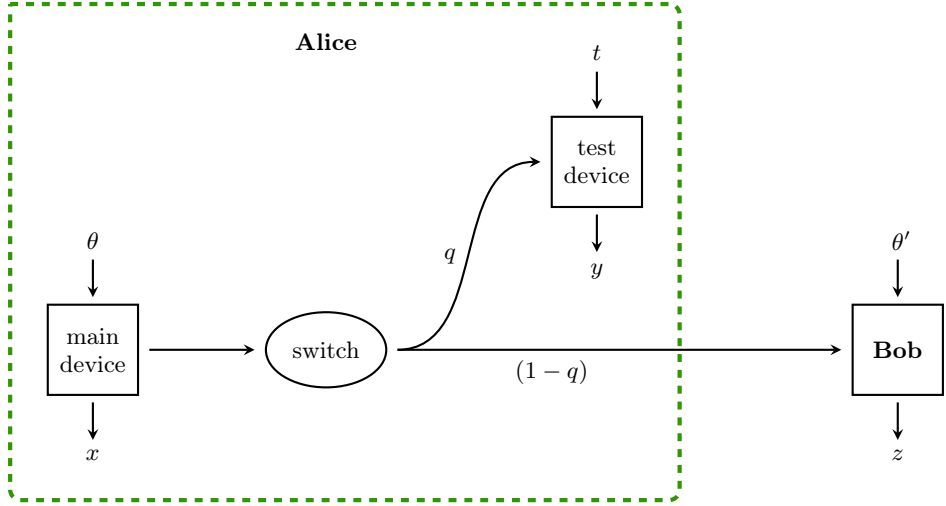
\begin{figure}[h!]
\begin{tikzpicture}[scale=1, line width=0.8]
	\corefig
	\draw [->] (7.6, 1.2) to (7.6, 0.7);
	\node[align=center] at (7.65, 1.45) {$\theta'$};
	\node[align=center] at (7.6, -1.4) {$z$};
\end{tikzpicture}
\caption{Honest execution of the DI WSE protocol. The main device prepares an EPR pair $\ket{\Psi_{AB}}$, measures the $A$ system in either the computational ($\theta = 0$) or Hadamard ($\theta = 1$) basis (chosen uniformly at random) to produce $x \in \{0, 1\}$, while the $B$ system is sent to the switch. Now, Alice chooses to either perform a test or play a live round. Whenever she decides to execute a test (with probability $q$), the switch directs $B$ to the test device, and she performs a CHSH test between the main device and the test device. That is, she chooses a random input $t \in \{0, 1\}$ and checks the CHSH condition $x \oplus y = \theta \cdot t$ on the outputs $x, y \in\{0, 1\}$. Whenever, she decides to play a live round (with probability $1-q$) she uses the switch to send $B$ to Bob, who measures the incoming qubit in either the computational ($\theta' = 0$) or Hadamard ($\theta' = 1$) basis (chosen uniformly at random) to produce $z \in \{0,1\}$, respectively. After $n$ live rounds, both parties wait time $\Delta t$, which enforces the storage assumption, after which Alice announces her basis string $\theta^{n} = \theta_{1} \theta_{2} \ldots \theta_{n}$. At the end Alice holds a random string $x^{n} = x_{1} x_{2} \ldots x_{n}$, while Bob has an index set $\mathcal{I} = \{j \in [n] : \theta_{j} = \theta'_{j}\}$ and a substring $x_{\mathcal{\cI}} := ( x_{j} )_{j \in \cI}$.}
\label{fig:honest} 
\end{figure}

We propose a DI protocol for WSE whose security is certified by the violation of the Clauser-Horne-Shimony-Holt (CHSH)~\cite{clauser69} inequality (see Section~\ref{sec:chsh} for details). We make the assumption that it is always the dishonest party that produced the devices. However, we will argue that dishonest Alice cannot gain any advantage by preparing Bob's devices so only the case of dishonest Bob requires detailed analysis. Before the protocol begins Bob provides Alice with two separate devices: a source of bipartite quantum states, combined with a measurement devices, plus one additional 
measurement devices that Alice can use for testing (see Fig.~\ref{fig:honest}). According to the ideal specification this setup should be capable of producing the maximal violation of the CHSH inequality. In the protocol, Alice will use a switch to either send a quantum state to the test device or to Bob. That is, she sometimes uses her devices to violate the CHSH inequality (the test rounds) while sometimes she only measures one of the particles and passes the other one to Bob (the live rounds). Intuitively, observing a high CHSH violation in the test rounds implies that measurements performed by the devices are incompatible, which leads to uncertainty (against a classical adversary) in the live rounds. For completeness, let us stress the importance of the assumption that Alice has full control over the switch, i.e.~she is free to choose which rounds are used for testing and which rounds are used in the protocol (sometimes referred to as the \emph{free will} assumption). This assumption is crucial from the theoretical point (it implies that the sample used to assess the performance of the devices cannot be influenced by the dishonest party, which is important since in many cases even limited influence may completely break the security), but it is also reasonable from a practical point of view (a switch is a simple enough device to be prepared by Alice herself).

In the dishonest scenario we allow Bob to prepare all the devices and in addition he receives quantum communication from Alice during the protocol as depicted in Fig.~\ref{fig:dishonest-Bob}. Here, we analyse two distinct security models.

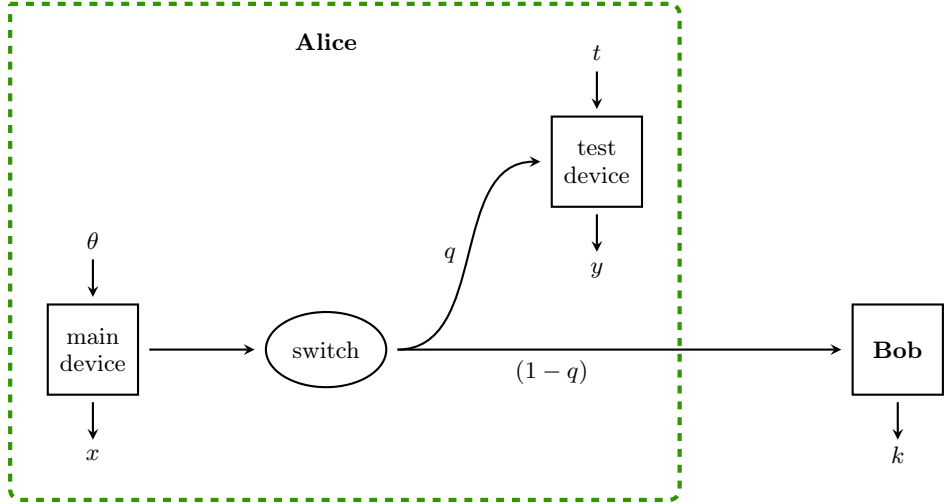
\begin{figure}[h!]
\begin{tikzpicture}[scale=1, line width=0.8]
	\corefig
	\node[align=center] at (7.6, -1.4) {$k$};
\end{tikzpicture}
\caption{Dishonest Bob prepared all the devices. This means that the state generated by the source can be chosen arbitrarily by Bob, and similarly he can adjust the measurements performed by the main and test device. Alice has control of the correctly functioning switch to decide whether she wants to test or perform a live round. Honest Alice proceeds as before, however, Bob is not restricted to performing BB84 measurements on the returning quantum states. In sharp contrast to DI QKD, the dishonest party thus receives quantum communication coming back from the devices, which calls for new techniques. As we will show below, it will be enough to consider the case where Bob measures the resulting quantum states to obtain some classical information. We will then establish a bound on the min-entropy that Bob has about the string $x^n$, given this classical 
information $k$ and the basis information received later. Using standard methods~\cite{konig12}, we can then turn this into a security statement in the noisy-storage model. 
}
\label{fig:dishonest-Bob}
\end{figure}
\begin{itemize}
\item Memoryless devices (against an arbitrary attack)\\
We call a device memoryless if its behaviour is identical every time it is used and there are no correlations between different uses. This is a convenient assumption because for such devices the observed CHSH violation $\beta$ is a well-defined quantity and can be estimated to arbitrary precision.
As explained in Fig.~\ref{fig:wse} the goal of WSE is to generate a string $X^{n}$ that Bob is at least partially ignorant about as quantified by the min-entropy $\hmin(X^{n} | \textnormal{Bob})$. In case of Bob whose quantum storage is restricted to be of dimension at most $d$ we show that
\begin{equation}
\label{eq:security-memoryless}
\hmin(X^{n} | \textnormal{Bob}) \geq n f(\beta) - \log d
\end{equation}
or equivalently
\begin{equation*}
\pguess(X^{n} | \textnormal{Bob}) \leq d \cdot 2^{- n f(\beta)},
\end{equation*}
where $f(\beta)$ is a simple function plotted in Fig.~\ref{fig:memoryless-bound} and $\log \equiv \log_{2}$. Thus, to achieve security against such an adversary it suffices to choose $n$ large enough to guarantee $n f(\beta) - \log d > 0$. For adversaries whose quantum storage is noisy rather than bounded the analysis is slightly more involved and can be found in Section~\ref{sec:memoryless} (explicit security bound in Proposition~\ref{prop:noisy}). In either case positive min-entropy rate implies that the protocol can be used for constructing more complicated primitives like bit commitment or oblivious transfer.
\begin{figure}
	\centering
	\includegraphics[scale=1]{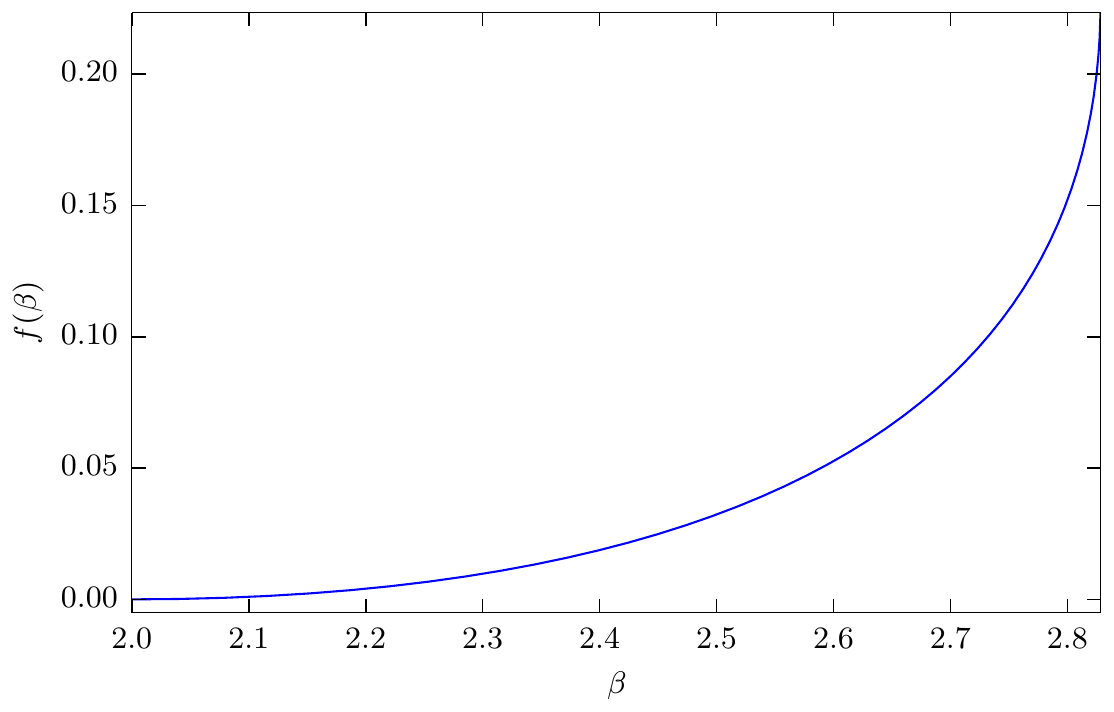}
	\caption{Lower bound on the min-entropy rate $f(\beta)$ as a function of the CHSH violation $\beta$. Crucially, we have
$f(\beta) > 0$, whenever $\beta > 2$. This means that security can be achieved for arbitrarily small violation of the CHSH inequality.}
	\label{fig:memoryless-bound}
\end{figure}
\item General devices against a sequential attack\\
In case of devices with memory (whose behaviour may change during the protocol and in particular there might be correlations between different rounds) the analysis is more involved both from the conceptual and technical point of view.
First, we must realise that we cannot in advance test the devices (to estimate their quality) and use the results to make a security statement simply because the behaviour of the devices might change in time. In particular, it is clear that the devices must not know whether they are currently  being tested or not. Therefore, the test rounds and the live rounds must be interspersed and we can only make a security statement about the combined performance.

In this case the test rounds must be explicitly included in the protocol and we adapt the simplest solution in which before every round Alice flips a biased coin and either plays a test round (with probability $q$) or a live round (with probability $1 - q$). After $n$ rounds she computes the fraction of successful CHSH rounds $\fchsh$ and checks whether it exceeds some previously chosen threshold 
$\gamma$. Note that estimating $\fchsh$ plays the role of estimating $\beta$ in the memoryless scenario: once the devices are allowed to have memory and change behaviour from round to round, $\beta$ is no longer a well-defined quantity and $\fchsh$ is the best approximation thereof. If $\fchsh \geq \gamma$ she declares the protocol to have terminated successfully, otherwise she aborts. Intuitively, what we want to avoid is the situation in which Alice believes that the protocol has terminated correctly but nevertheless Bob actually knows the entire string $x^{n}$ and we denote such an event by $F$ (failure). Suppose $n$ rounds are executed with parameters $q \in [0, 1]$ and $\gamma \in [\frac{3}{4}, 1]$. 
We call an attack sequential if after every round Bob is required to produce a classical outcome and his guess for that round is required to be a (classical) post-processing of that outcome combined with the basis information and any information from the previous rounds (see Section~\ref{sec:general} for a more detailed explanation). It is worth noting that this assumption removes the need to restrict Bob's storage capabilities: since he is forced to commit to his guess immediately after the round is over, storing the quantum system does not help). We show that in the sequential scenario the probability of failure is bounded by
\begin{equation}
\label{eq:security-general}
\Pr[ F ] \leq [ \alpha_{\textnormal{min}}(q, \gamma) ]^{n},
\end{equation}
where $\alpha_{\textnormal{min}}(q, \gamma)$ can be easily calculated for any (valid) choice of $q$ and $\gamma$ (cf.~Fig.~\ref{fig:alphamin}). Alternatively, we can write $\Pr[F]$ in terms of the probability of passing the test $p_{\textnormal{pass}}$ and the probability of successfully guessing the entire ``live'' string (restricted to sequential guessing strategies, see Section~\ref{sec:guessing} for a precise definition) conditioned on passing the test $\pguessseq( X^{\cL} | \textnormal{Bob}, \textnormal{pass} )$
\begin{equation}
\label{eq:security-general-2}
\Pr[ F ] = p_{\textnormal{pass}} \cdot \pguessseq( X^{\cL} | \textnormal{Bob}, \textnormal{pass} ) \leq [ \alpha_{\textnormal{min}}(q, \gamma) ]^{n}.
\end{equation}
Our analysis is tight in the sense that it identifies correctly the pairs $(q, \gamma)$ for which security is possible, i.e.~we show that $\alpha_{\textnormal{min}}(q, \gamma) < 1$ unless $q = 0$ (Alice never tests), $q = 1$ (Alice never plays a live round) or $\gamma = \frac{3}{4}$ (the threshold can be achieved by a classical strategy). This means that the probability of the devices performing well in the test rounds \emph{and} failing to implement a secure WSE decays exponentially in the total number of rounds. The technique we use to prove this result is generic and can be applied to any situation in which the combined performance of two (or more) games is assessed (as long as there is some non-trivial trade-off between them).
\begin{figure}
	\centering
	\includegraphics[scale = 1]{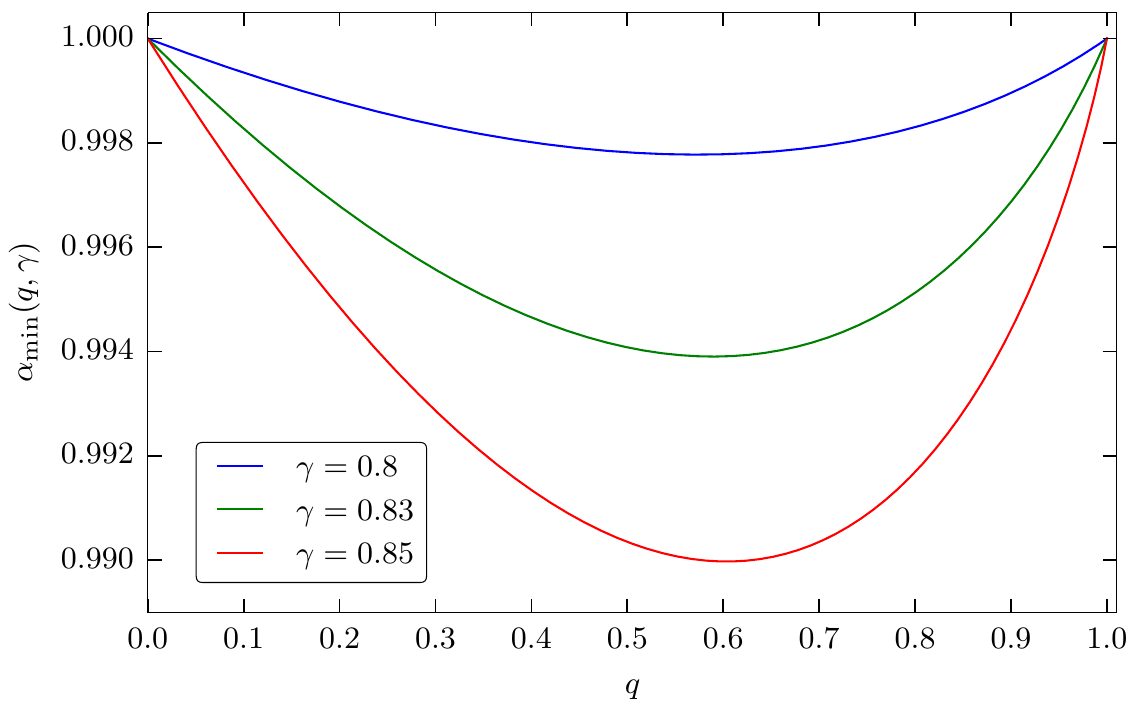}
	\caption{Values of the decay rate $\alpha_{\textnormal{min}}(q, \gamma)$ calculated numerically as a function of $q$ for various values of $\gamma$.}
	\label{fig:alphamin}
\end{figure}
\end{itemize}
These two contributions should be seen as steps towards a security proof against the most general attack. The memoryless model might be of independent interest since it captures the case of devices which are faulty rather than malicious (e.g.~due to some misalignment of optical components); such scenarios are usually modelled as permanent deviations from the ideal specification rather than time-dependent ones.
\section{Methods}
In Section~\ref{sec:wse-protocol} we present the original protocol for WSE using trusted devices, in Section~\ref{sec:preliminaries} we introduce the relevant quantities and prove some technical lemmas, in Section~\ref{sec:memoryless} we formalise the scenario of memoryless devices and prove security statement~\eqref{eq:security-memoryless} and in Section~\ref{sec:general} we analyse the case of arbitrary devices against sequential attacks and prove security claim~\eqref{eq:security-general}.

\subsection{The original WSE protocol for trusted devices}
\label{sec:wse-protocol}
To build intuition, let us first describe the original protocol for WSE~\cite{konig12}, which works under the assumption that the devices used by Alice and Bob are perfect and prepared in a trustworthy fashion. We sketch out a simple security argument and discuss how to make the protocol device-independent. Note that there exist more sophisticated arguments which give better security guarantees but they seem to be more difficult to adapt to the DI scenario.
\begin{prot}{WSE in the noisy-storage model}
\label{prot:original}
\begin{enumerate}[label=(\arabic*)]
\item Alice chooses two uniform $n$-bit strings $x^{n}, \theta^{n} \in \{0, 1\}^{n}$, generates the $n$-qubit state
\begin{equation*}
\bigotimes_{j = 1}^{n} H^{\theta_{j}} \ket{x_{j}},
\end{equation*}
where $H$ is the Hadamard gate, and sends it to Bob. (Note that this just a sequence of $n$ randomly chosen BB84~\cite{bennett84} states.)
\item Bob chooses a uniform $n$-bit string $\theta'^{n} \in \{0, 1\}^{n}$ and measures the $j$-th qubit in the computational (if $\theta_{j}' = 0$) or Hadamard (if $\theta'_{j} = 1$) basis.
\item Alice waits a fixed amount of time (to enforce the restriction on Bob's quantum memory) and then sends $\theta^{n}$ to Bob.
\item Bob determines the index set as
\begin{equation*}
\cI := \{ j \in [n] : \theta_{j} = \theta_{j}' \}
\end{equation*}
and obtains the corresponding substring $x_{\cI}$.
\end{enumerate}
\end{prot}
Correctness of this protocol is easy to verify because the string $x^{n}$ is chosen uniformly at random by Alice and with high probability Bob measures roughly half of the qubits in the correct basis. Security for honest Bob is a direct consequence of the fact that the index set is determined by the positions at which $\theta_{j} \oplus \theta_{j}' = 0$. Since $\theta'^{n}$ is chosen uniformly at random by Bob, every index set is equally likely (and Alice is fully ignorant about it). Therefore, the only non-trivial scenario is the case of honest Alice.

Let $\rho_{X^{n} \Theta^{n} B}$ be the state of the protocol after step (1), where $X^{n}$ and $\Theta^{n}$ are the classical random variables generated by Alice and $B$ is the quantum system received by Bob. The memory bound forces Bob to put the $B$ subsystem through a quantum channel which outputs a classical register $K$ and a quantum register $Q$, which gives rise to $\rho_{X^{n} \Theta^{n} K Q}$. Since $\Theta^{n}$ is eventually announced to Bob, our goal is to find a lower bound on $\hmin(X^{n} | K Q \Theta^{n} )$. In the bounded-storage model we can use the following chain rule
\begin{equation}
\label{eq:chain-rule}
\hmin(X^{n} | K Q \Theta^{n} ) \geq \hmin(X^{n} | K \Theta^{n} ) - \log \dim Q.
\end{equation}
In case of noisy storage the argument is slightly more involved (see Section~\ref{sec:memoryless} for details) but again the task reduces to establishing uncertainty against a classical adversary. This is possible because generating random BB84 states is equivalent to creating EPR pairs and measuring them in either computational or Hadamard basis and we know that outcomes of incompatible measurements cannot be predicted (perfectly) by a classical adversary. Indeed, it has been shown (Eq.~(18) in Ref.~\cite{konig12}) that the resulting conditional min-entropy satisfies $\hmin(X^{n} | K \Theta^{n} ) \geq \alpha n$ for
\begin{equation*}
\alpha = - \log \Big( \frac{1}{2} + \frac{1}{2 \sqrt{2}} \Big) \approx 0.22.
\end{equation*}
Note that this bound is tight and is achieved if Bob measures every received qubit in the intermediate basis $\{ \ket{\alpha_{0}}, \ket{\alpha_{1}} \}$, where
\begin{gather*}
\ket{\alpha_{0}} = \cos (\pi/8) \ket{0} + \sin (\pi/8) \ket{1},\\
\ket{\alpha_{1}} = \sin (\pi/8) \ket{0} - \cos (\pi/8) \ket{1}.
\end{gather*}

In case of trusted devices placing a lower bound on $\hmin(X^{n} | K \Theta^{n} )$ is possible because we know exactly the measurement operators on Alice's side. The main challenge in the DI scenario is to prove a lower bound which relies solely on properties that can be certified device-independently. Our approach follows the intuition that observing a Bell violation implies incompatibility of local observables which is sufficient to guarantee uncertainty. Previously, this approach has been used successfully in proving security of DI QKD~\cite{tomamichel13a, lim13}.

\subsection{Preliminaries}
\label{sec:preliminaries}
For an integer $n \in \amsbb{N}$ let $[n] := \{1, 2, \ldots, n\}$. Throughout this paper we assume that all random variables are discrete (they take a finite number of values) and that all quantum systems are finite-dimensional. Let $\sH$ be a (finite-dimensional) Hilbert space and let $\cL(\sH)$/$\cH(\sH)$ be the set of linear/Hermitian operators acting on $\sH$. The Schatten $\infty$-norm of an operator $X$ is denoted by $\norm{X}$. The square root of a positive semidefinite operator $X$, denoted by $\sqrt{X}$, is defined as the unique positive semidefinite operator $Y$ satisfying $Y^{2} = X$. The modulus of an operator $X$, denoted by $\abs{X}$, is defined as $Y = \sqrt{ X^{\dagger} X }$. It is easy to verify that for arbitrary operators $X$ and $Y$ we have
\begin{equation}
\label{eq:anticom-com-equality}
\abs{ X + Y }^{2} + \abs{ X - Y }^{2} = 2 ( X^{\dagger} X + Y^{\dagger} Y ).
\end{equation}
The commutator of $X$ and $Y$ is defined as $[X, Y] = XY - YX$, while the anticommutator is defined as $\{ X, Y \} = XY + YX$.

A quantum state $\rho$ is a Hermitian operator $\rho \in \cH(\sH)$ which is positive semidefinite ($\rho \geq 0$) and of unit trace ($\tr \rho = 1$). An observable is a Hermitian operator $A \in \cH(\sH)$ which satisfies $- \mathbb{1} \leq A \leq \mathbb{1}$ (or equivalently $\norm{A} \leq 1$). Plugging $X = AB$ and $Y = BA$ into Eq.~\eqref{eq:anticom-com-equality} gives
\begin{equation}
\label{eq:operator-inequality}
\abs{ \{A, B\} }^{2} + \abs{ [A, B] }^{2} = 2 ( A B^{2} A + B A^{2} B ) \leq 4 \cdot \mathbb{1},
\end{equation}
where the upper bound follows from the fact that $A^{2}, B^{2} \leq \mathbb{1}$.
\subsubsection{The CHSH inequality}
\label{sec:chsh}
In 1964 John Bell showed that measuring quantum systems leads to stronger-than-classical correlations~\cite{bell64}. In 1969 Clauser, Horne, Shimony and Holt spelt out the simplest scenario in which this can be observed \cite{clauser69}. Let $\sH_{A}$ and $\sH_{B}$ be Hilbert spaces and let $A_{0}, A_{1} \in \cH(\sH_{A})$ and $B_{0}, B_{1} \in \cH(\sH_{B})$ be binary observables. The CHSH operator is defined as
\begin{equation*}
W = A_{0} \otimes B_{0} + A_{0} \otimes B_{1} + A_{1} \otimes B_{0} - A_{1} \otimes B_{1}
\end{equation*}
and the CHSH value equals $\beta = \tr (W \rho_{AB})$, where $\rho_{AB}$ is a bipartite quantum state on $\sH_{A} \otimes \sH_{B}$. It is known that there exist a state and observables that yield $\beta = 2 \sqrt{2}$. On the other hand, if we restrict ourselves to classical systems (which can be enforced by requiring the observables to commute, i.e.~$[A_{0}, A_{1}] = [B_{0}, B_{1}] = 0$) we can only reach $\beta = 2$. This scenario can be equivalently cast as a two-player game in which Alice receives $x$, Bob receives $y$ (both chosen uniformly at random) and are required to output $a$ and $b$, respectively. The game is won if $a \oplus b = x \cdot y$ and it is straightforward to show that the winning probability of this game $\pwin$ and the CHSH value $\beta$ are related by
\begin{equation*}
\pwin = \frac{1}{2} + \frac{\beta}{8}.
\end{equation*}
Therefore, the optimal classical winning probability equals $\frac{3}{4}$, while the optimal quantum winning probability equals $\frac{1}{2} + \frac{1}{2 \sqrt{2}} \approx 0.85$.
\subsubsection{Guessing with postmeasurement information}
\label{sec:guessing}
We start by defining the guessing probability and min-entropy for a classical-quantum (cq) state (we denote the quantum register by $B$ to be consistent with the protocol in which it is the dishonest Bob who faces the task of guessing).
\begin{df}
Let $\rho_{XB}$ be a cq-state
\begin{equation*}
\rho_{XB} = \sum_{x} p_{x} \ketbraq{x} \otimes \rho_{x}^{B},
\end{equation*}
where $\rho_{x}^{B}$ are (normalised) quantum states and $\sum_{x} p_{x} = 1$. The optimal guessing probability of $X$ given access to $B$ is defined as
\begin{equation*}
\pguess(X | B) := \max_{ \{ M_{x} \}_{x} } \sum_{x} p_{x} \cdot \tr ( M_{x} \rho_{x}^{B} ),
\end{equation*}
where the maximisation is taken over all POVMs. The conditional min-entropy of $X$ given $B$ is defined as
\begin{equation*}
\hmin(X | B) := - \log \pguess(X | B).
\end{equation*}
\end{df}
\noindent Note that computing the guessing probability can be written as a semidefinite program, i.e.~it can be computed efficiently (in the input dimension). For a classical probability distribution $P_{XY}$ the expression simplifies to
\begin{equation*}
\pguess(X | Y) = \sum_{y} P_{Y} (y) \cdot \max_{x} P_{X | Y}(x | y).
\end{equation*}
Alternatively, this maximisation can be written more compactly as
\begin{equation*}
\pguess(X | Y) = \max_{f} \Pr[X = f(Y)],
\end{equation*}
where the maximisation is taken over deterministic functions $f : \cY \to \cX$. It can be shown~\cite{wehner08b} that the min-entropy is additive on tensor products, i.e.~given two uncorrelated cq-states $\rho_{X_{1} B_{1}} \otimes \rho_{X_{2} B_{2}}$ we have
\begin{equation*}
\hmin(X_{1} X_{2} | B_{1} B_{2}) = \hmin(X_{1} | B_{1}) + \hmin(X_{2} | B_{2}).
\end{equation*}
We also need the notion of smooth min-entropy.
\begin{df}
For $\varepsilon \geq 0$ let $\cB^{\varepsilon}(\rho_{XB})$ be the ball of cq-states of radius $\varepsilon$ around $\rho_{XB}$, i.e.~$\sigma_{XB} \in \cB^{\varepsilon}(\rho_{XB})$ iff $\sigma_{XB}$ is a cq-state and
\begin{equation*}
\frac{1}{2} \norm{ \sigma_{XB} - \rho_{XB} }_{1} \leq \varepsilon,
\end{equation*}
where $\norm{\cdot}_{1}$ denotes the trace norm (Schatten 1-norm). Then, the smooth min-entropy of a cq-state $\rho_{XB}$ is defined as
\begin{equation*}
\hmineps(X | B)_{\rho} := \sup_{\sigma_{XB} \in \cB^{\varepsilon}(\rho_{XB})} \hmin(X | B)_{\sigma}.
\end{equation*}
\end{df}

Security analysis of two-party cryptography in the bounded or noisy storage model leads to the task of \emph{guessing with postmeasurement information} originally considered by Ballester, Wehner and Winter~\cite{ballester08}. Let $\rho_{XY\hspace{-1pt}B}$ be a tripartite ccq-state, where $X$ is a classical register taking values in $\cX$, $Y$ is a classical register taking values in $\cY$ and $B$ is the quantum system of Bob. In the postmeasurement information scenario Bob is forced to measure his subsystem $B$ to obtain some classical information $F$ before learning $Y$. Later he learns the postmeasurement information $Y$ and must produce a guess for $X$. We will later show that without loss of generality we can assume that the outcomes of Bob's measurement (i.e.~the possible values of $F$) are labelled by functions $f : \cY \to \cX$ such that Bob's optimal guess upon receiving $y$ is $f(y)$. Equivalently we can think of the outcome of the measurement as a sequence of guesses: one for \emph{every possible value} of the postmeasurement information.
\begin{df}
Let $\rho_{XY\hspace{-1pt}B}$ be a ccq-state
\begin{equation*}
\rho_{XY\hspace{-1pt}B} = \sum_{xy} p_{xy} \ketbraq{x} \otimes \ketbraq{y} \otimes \rho_{xy}^{B}.
\end{equation*}
The optimal guessing probability of $X$ given access to $B$ with $Y$ as postmeasurement information is defined as
\begin{equation*}
\pguess( X | B Y\post ) := \max_{ \{ M_{f} \}_{f} } \sum_{\substack{x, y, f\\x = f(y)}} p_{xy} \cdot \tr ( M_{f} \rho_{xy}^{B} ),
\end{equation*}
where the maximisation is taken over all POVMs with $\abs{\cX}^{\abs{\cY}}$ outcomes labelled by functions $f : \cY \to \cX$ and the star (\hspace{2pt}$^{*}$\hspace{-1pt}) indicates that $Y$ is only available \emph{after} the measurement. The conditional min-entropy of $X$ given $B$ with $Y$ as postmeasurement information is defined as
\begin{equation*}
\hmin(X | B Y\post) := - \log \pguess(X | B Y\post).
\end{equation*}
\end{df}
\noindent This is a useful formulation because defining
\begin{equation*}
\sigma_{f}^{B} = \sum_{\substack{x, y\\x = f(y)}} p_{xy} \rho_{xy}^{B}
\end{equation*}
allows us to rewrite the objective function as
\begin{equation*}
\sum_{\substack{x, y, f\\x = f(y)}} p_{xy} \cdot \tr ( M_{f} \rho_{xy}^{B} ) = \sum_{f} \tr (M_{f} \sigma_{f}^{B}),
\end{equation*}
which is equivalent to the standard guessing probability $\pguess(F | B)$ for the (unnormalised) state
\begin{equation*}
\rho_{FB} = \sum_{f} \ketbraq{f} \otimes \sigma_{f}^{B}.
\end{equation*}
Therefore, this problem can also be solved efficiently using semidefinite programming techniques~\cite{ballester08}. Moreover, just like in the standard guessing scenario, the min-entropy is additive over tensor products, i.e.~given two uncorrelated ccq-states $\rho_{X_{1} Y_{1} B_{1}} \otimes \rho_{X_{2} Y_{2} B_{2}}$ we have
\begin{equation}
\label{eq:additivity}
\hmin(X_{1} X_{2} | B_{1} B_{2} Y_{1}\post Y_{2}\post) = \hmin(X_{1} | B_{1} Y_{1}\post) + \hmin(X_{2} | B_{2} Y_{2}\post).
\end{equation}
The following proposition gives an alternative (but equivalent) formulation of the min-entropy with postmeasurement information.
\begin{prop}
\label{prop:equivalence}
Let $\rho_{XY\hspace{-1pt}B}$ be a ccq-state and let $\cP$ be the set of tripartite probability distributions over $X, Y$ and $K$ which can be obtained by measuring subsystem $B$, i.e.~$P_{XY\hspace{-1pt}K} \in \cP$ iff there exists a measurement $\{N_{k}\}_{k}$ such that
\begin{equation*}
\Pr[ X = x, Y = y, K = k] = p_{xy} \cdot \tr (N_{k} \rho_{xy}^{B}).
\end{equation*}
Then, the following relation holds
\begin{equation}
\label{eq:equivalence}
\pguess( X | B Y\post ) = \sup_{ P_{XY\hspace{-1pt}K} \in \cP } \pguess(X | K Y).
\end{equation}
\end{prop}
\begin{proof}
Let us first show that the left-hand side is never larger than the right-hand side. Let $\{M_{f}\}_{f}$ be the POVM which saturates the left-hand side and let $P_{XY\hspace{-1pt}F}$ be the resulting probability distribution. Then
\begin{align*}
\pguess( X | B Y\post ) &= \sum_{\substack{x, y, f\\x = f(y)}} p_{xy} \cdot \tr ( M_{f} \rho_{xy}^{B} ) = \sum_{\substack{x, y, f\\x = f(y)}} P_{XY\hspace{-1.25pt}F}(x y f) = \sum_{y, f} P_{Y\hspace{-1.25pt}F}(y f) \cdot P_{X | Y\hspace{-1.25pt}F}( f(y) | y f)\\
&\leq \sum_{y, f} P_{Y\hspace{-1.25pt}F}(y f) \cdot \max_{x} P_{X | Y\hspace{-1.25pt}F}( x | y f) = \pguess(X | F Y) \leq \sup_{ P_{XY\hspace{-1pt}K} \in \cP } \pguess(X | K Y).
\end{align*}
To prove the other direction consider an arbitrary measurement $\{N_{k}\}_{k}$ (with a finite number of outcomes) which leads to the probability distribution $P_{X YK}$. For every value of $k$ we define a function $g_{k} : \cY \to \cX$ such that
\begin{equation*}
g_{k}(y) = \argmax_{x} P_{X | Y\hspace{-1.25pt}K}(x | yk).
\end{equation*}
This allows us construct a new measurement whose outcomes are labelled by functions $f : \cY \to \cX$
\begin{equation*}
M_{f} = \sum_{k : g_{k} = f} N_{k}.
\end{equation*}
Using this measurement gives
\begin{align*}
\pguess( X | B Y\post ) &\geq \sum_{\substack{x, y, f\\x = f(y)}} p_{xy} \cdot \tr ( M_{f} \rho_{xy} ) = \sum_{\substack{x, y, f\\x = f(y)}} \sum_{k : g_{k} = f} p_{xy} \cdot \tr ( N_{k} \rho_{xy} ) = \sum_{\substack{x, y, k\\x = g_{k}(y)}} P_{XY\hspace{-1pt}K}(xyk)\\
&= \sum_{y, k} P_{Y\hspace{-1pt}K}(yk) \cdot P_{X | Y\hspace{-1pt}K}( g_{k}(y) | ky) = \sum_{y, k} P_{Y\hspace{-1pt}K}(yk) \cdot \max_{x} P_{X | Y\hspace{-1pt}K}( x | yk) = \pguess(X | K Y).
\end{align*}
By considering measurements that approach the optimal guessing probability we conclude that Eq.~\eqref{eq:equivalence} holds. In particular, this implies that the supremum can be replaced by a maximum.
\end{proof}
The final security statement in the scenario of devices with memory is phrased in terms of \emph{sequential guessing probability}. Intuitively, this corresponds to the situation in which Bob is required to guess a sequence of random variables but before each guess he gains access to an extra ``advice variable''.
\begin{df}
Let $P_{X_{1} X_{2} \ldots X_{n} Y_{1} Y_{2} \ldots Y_{n}}$ be a probability distribution of $2n$ variables, where $X_{j}$ and $Y_{j}$ take values in some arbitrary finite sets $\cX$ and $\cY$, respectively. The sequential guessing probability of $X^{n} = X_{1} X_{2} \ldots X_{n}$ given $Y^{n} = Y_{1} Y_{2} \ldots Y_{n}$ is defined as
\begin{equation*}
\pguessseq( X^{n} | Y^{n} ) = \max_{ \{f_{j}\}_{j} } \Pr[ \bigwedge_{j = 1}^{n} X_{j} = f_{j}( Y_{1} Y_{2} \ldots Y_{j} ) ],
\end{equation*}
where the maximisation is taken over deterministic functions $\{f_{j}\}_{j}$ such that $f_{j} : \cY^{\times j} \to \cX$.
\end{df}
The sequential character of this quantity makes it meaningful to talk about a subset of rounds, e.g.~the probability of successfully guessing the first $j$ variables $\pguessseq( X^{j} | Y^{j} )$ is a well-defined quantity that depends only on $P_{X^{j} Y^{j}}$. This stands in contrast to the usual guessing probability in which evaluating the probability of successfully guessing the first bit requires the knowledge of the complete set of ``advice variables''. Thanks to this property the sequential guessing probability behaves well under conditioning
\begin{equation*}
\pguessseq( X^{n} | Y^{n} ) = \pguessseq( X^{n - 1} | Y^{n - 1} ) \cdot \pguess ( X_{n} | Y^{n}, \cS ),
\end{equation*}
where the second term is just the standard guessing probability of the last bit \emph{conditional} on event $\cS$, which corresponds to (sequentially) guessing the first $n - 1$ bits correctly.
\subsubsection{Relation between transmitting classical information and uncertainty against noisy storage}
Let $\cF : \cL( \sH_{Q_{\textnormal{in}}} ) \to \cL( \sH_{Q_{\textnormal{out}}} )$ be a quantum channel (a completely positive, trace preserving map) and suppose we want to use it to transmit $k$ bits of information. The following definition captures how well this can be achieved.
\begin{df}
The optimal probability of successfully transmitting $k$ bits of information through the channel $\cF$ is defined as
\begin{equation*}
\Psucc(k) = \max_{ \{ \rho_{x} \}_{x}, \{ M_{x} \}_{x} } \frac{1}{2^{k}} \sum_{x \in \{0, 1\}^{k}} \tr [ M_{x} \cF(\rho_{x}) ],
\end{equation*}
where $\{ \rho_{x} \}_{x}$ represents the encoding procedure (a set of $2^{k}$ normalised states on $Q_{\textnormal{in}}$) while $\{ M_{x} \}_{x}$ is the decoding measurement (a measurement on $Q_{\textnormal{out}}$ with $2^{k}$ outcomes).
\end{df}
\noindent The following lemma by K\"{o}nig, Wehner and Wullschleger relates the success probability to the maximal decrease in entropy in the noisy storage setting~\cite{konig12}.
\begin{lem}[Lemma II.2,~\cite{konig12}]
\label{lem:uncertainty-transmission}
Let $\cF : \cL( \sH_{Q_{\textnormal{in}}} ) \to \cL( \sH_{Q_{\textnormal{out}}} )$ be a CPTP map. Consider an arbitrary ccq-state $\rho_{XTQ}$ and define
\begin{equation*}
\sigma_{XT Q_{\textnormal{out}}} := ( \id_{XT} \otimes \cF_{Q_{\textnormal{in}} \to Q_{\textnormal{out}}}) (\rho_{XTQ_{\textnormal{in}}}),
\end{equation*}
where $\id$ stands for the identity channel. For any $\varepsilon > 0$ we have
\begin{equation*}
\hmineps( X | T Q_{\textnormal{out}} )_{\sigma} \geq - \log \Psucc \big( \lfloor \hmin(X | T) - \log (1/\varepsilon) \rfloor \big).
\end{equation*}
\end{lem}
\subsubsection{Trade-off between non-locality and uncertainty against classical adversaries}
As mentioned before a crucial component of our analysis is the trade-off between how well a pair of devices can perform in the CHSH test and how unpredictable the output of a single device is against a classical adversary. It turns out that such a (tight) trade-off can be established by finding the right measure of incompatibility of binary observables. In our previous work we have used the effective anticommutator as a measure of incompatibility~\cite{kaniewski14}. Unfortunately, this quantity does not allow us to bound uncertainty against classical side information (see Appendix \ref{app:effective-anticommutation-insufficient} for a counterexample) so here we consider a more refined quantity: the \emph{absolute effective anticommutator}. Proposition~\ref{prop:chsh-anticommutator} shows that observing a CHSH violation places an upper bound on the absolute effective anticommutator.
\begin{prop}
\label{prop:chsh-anticommutator}
Let $\rho_{AB} \in \cH(\sH_{A} \otimes \sH_{B})$ be a bipartite quantum state and let $A_{0}, A_{1} \in \cH(\sH_{A})$ and $B_{0}, B_{1} \in \cH(\sH_{B})$ be observables. The absolute effective anticommutator on Alice's side is defined as
\begin{equation*}
\varepsilon_{+} := \frac{1}{2} \tr \big( \abs{ \{A_{0}, A_{1} \} } \rho_{A} \big).
\end{equation*}
The CHSH value of the setup is defined as $\beta := \tr ( W \rho_{AB} )$ for
\begin{equation*}
W = A_{0} \otimes B_{0} + A_{0} \otimes B_{1} + A_{1} \otimes B_{0} - A_{1} \otimes B_{1}.
\end{equation*}
The following relation holds
\begin{equation}
\label{eq:beta-absolute-effective-anticommutator}
\abs{\beta} \leq 2 \sqrt{ 1 + \sqrt{ 1 - \varepsilon_{+}^{2} } }.
\end{equation}
\end{prop}
\begin{proof}
The proof is a sequence of elementary inequalities (either at the level of numbers or operators). We will repeatedly use the Cauchy-Schwarz inequality, which says that for arbitrary operators $X$ and $Y$ we have
\begin{equation*}
\abs{ \tr( X^{\dagger} Y ) }^{2} \leq \tr( X^{\dagger} X ) \cdot \tr (Y^{\dagger} Y).
\end{equation*}
We start by setting $X^{\dagger} = W \sqrt{\rho_{AB}}$ and $Y = \sqrt{\rho_{AB}}$ which gives
\begin{equation}
\label{eq:beta-squared}
\beta^{2} = \big[ \tr ( W \rho_{AB} ) \big]^{2} \leq \tr (W^{2} \rho_{AB}).
\end{equation}
Writing out $W^{2}$ explicitly gives
\begin{equation*}
W^{2} = A_{0}^{2} \otimes (B_{0} + B_{1})^{2} + A_{1}^{2} \otimes (B_{0} - B_{1})^{2} + \{A_{0}, A_{1}\} \otimes ( B_{0}^{2} - B_{1}^{2} ) - [A_{0}, A_{1}] \otimes [B_{0}, B_{1}].
\end{equation*}
Let us first focus on the first three terms. Upperbounding $A_{0}^{2}$ and $A_{1}^{2}$ by $\mathbb{1}$ gives
\begin{equation*}
A_{0}^{2} \otimes (B_{0} + B_{1})^{2} + A_{1}^{2} \otimes (B_{0} - B_{1})^{2} + \{A_{0}, A_{1}\} \otimes ( B_{0}^{2} - B_{1}^{2} ) \leq \mathbb{1} \otimes 2 ( B_{0}^{2} + B_{1}^{2} ) + \{A_{0}, A_{1}\} \otimes ( B_{0}^{2} - B_{1}^{2} ).
\end{equation*}
Writing the identity in the eigenbasis of the anticommutator $\{ A_{0}, A_{1} \} = \sum_{k} \lambda_{k} \ketbraq{e_{k}}$ gives
\begin{equation*}
\mathbb{1} \otimes 2 ( B_{0}^{2} + B_{1}^{2} ) + \{A_{0}, A_{1}\} \otimes ( B_{0}^{2} - B_{1}^{2} ) = \sum_{k} \ketbraq{e_{k}} \otimes \big[ (2 + \lambda_{k}) B_{0}^{2} + (2 - \lambda_{k}) B_{1}^{2} \big] \leq 4 \cdot \mathbb{1} \otimes \mathbb{1},
\end{equation*}
where the last inequality comes from upperbounding $B_{0}^{2}$ and $B_{1}^{2}$ by $\mathbb{1}$ (note that $\abs{\lambda_{k}} \leq 2$). We have therefore established that
\begin{equation*}
W^{2} \leq 4 \cdot \mathbb{1} \otimes \mathbb{1} + \big( - [A_{0}, A_{1}] \otimes [B_{0}, B_{1}] \big).
\end{equation*}
We bound the second term by its (operator) modulus
\begin{equation*}
- [A_{0}, A_{1}] \otimes [B_{0}, B_{1}] \leq \abs[\big]{[A_{0}, A_{1}] \otimes [B_{0}, B_{1}]} = \abs{[A_{0}, A_{1}]} \otimes \abs{[B_{0}, B_{1}]}.
\end{equation*}
Neglecting the anticommutator term in inequality~\eqref{eq:operator-inequality} leads to
\begin{equation*}
\abs{[B_{0}, B_{1}]}^{2} \leq 4 \cdot \mathbb{1},
\end{equation*}
which implies that $\abs{[B_{0}, B_{1}]} \leq 2 \cdot \mathbb{1}$. Therefore,
\begin{equation*}
W^{2} \leq 4 \cdot \mathbb{1} \otimes \mathbb{1} + 2 \, \abs{ [A_{0}, A_{1}] } \otimes \mathbb{1}
\end{equation*}
and
\begin{equation}
\label{eq:w-squared}
\tr( W^{2} \rho_{AB} ) \leq 4 + 2 \tr \big( \abs{ [A_{0}, A_{1}] } \rho_{A} \big).
\end{equation}
To upperbound $\tr \big( \abs{ [A_{0}, A_{1}] } \rho_{A} \big)$ we again use the Cauchy-Schwarz inequality with $X^{\dagger} = \abs{ [A_{0}, A_{1}] } \sqrt{\rho_{A}}$ and $Y = \sqrt{\rho_{A}}$ which gives
\begin{equation}
\label{eq:abs-commutator}
\big[ \tr \big( \abs{ [A_{0}, A_{1}] } \rho_{A} \big) \big]^{2} \leq \tr \big( \abs{ [A_{0}, A_{1}] }^{2} \rho_{A} \big).
\end{equation}
Inequality~\eqref{eq:operator-inequality} implies that
\begin{equation}
\label{eq:abs-commutator-squared}
\tr \big( \abs{ [A_{0}, A_{1}] }^{2} \rho_{A} \big) \leq 4 - \tr \big( \abs{\{A_{0}, A_{1} \}}^{2} \rho_{A} \big).
\end{equation}
Using the Cauchy-Schwarz inequality one last time with $X^{\dagger} = \abs{ \{A_{0}, A_{1}\} } \sqrt{\rho_{A}}$ and $Y = \sqrt{\rho_{A}}$ gives
\begin{equation}
\label{eq:abs-anticommutator-squared}
\big[ \tr \big( \abs{ \{ A_{0}, A_{1} \} } \rho_{A} \big) \big]^{2} \leq \tr \big( \abs{ \{ A_{0}, A_{1} \} }^{2} \rho_{A} \big).
\end{equation}
Since the left-hand side of Eq.~\eqref{eq:abs-anticommutator-squared} equals $4 \varepsilon_{+}^{2}$ combining it with inequalities~\eqref{eq:beta-squared}, \eqref{eq:w-squared}, \eqref{eq:abs-commutator} and \eqref{eq:abs-commutator-squared} gives
\begin{equation*}
\beta^{2} \leq 4 \Big( 1 + \sqrt{ 1 - \varepsilon_{+}^{2} } \, \Big).
\end{equation*}
Taking a square root leads to the desired result.
\end{proof}
\noindent It is easy to verify that this relation is in fact tight (it suffices to consider projective rank-1 measurements on the maximally entangled state of two qubits). In Proposition~\ref{prop:uncertainty-anticommutator} we show that the absolute effective anticommutator being small implies uncertainty against classical adversaries.
\begin{prop}
\label{prop:uncertainty-anticommutator}
Let $\rho_{AK}$ be a quantum-classical state
\begin{equation*}
\rho_{AK} = \sum_{k} p_{k} \rho_{k}^{A} \otimes \ketbraq{k}
\end{equation*}
and let $A_{0}$ and $A_{1}$ be two observables acting on the register $A$. Let $\varepsilon_{+} = \frac{1}{2} \tr \big( \abs{ \{ A_{0}, A_{1} \} } \rho_{A} \big)$ for $\rho_{A} = \sum_{k} p_{k} \rho_{k}^{A}$. Measuring the observable chosen by a uniformly random register $\Theta$ and storing the outcome in the register $X$ leads to the following probability distribution.
\begin{equation*}
\Pr[ X = x, \Theta = \theta, K = k] = \frac{1}{2} \cdot p_{k} \cdot \frac{1 + \tr (A_{\theta} \rho_{k}^{A}) }{2}.
\end{equation*}
Then, the guessing probability satisfies
\begin{equation}
\label{eq:pguess-absolute-effective-anticommutator}
\pguess(X | K \Theta) \leq \frac{1}{2} + \frac{1}{2} \sqrt{ \frac{ 1 + \varepsilon_{+} }{2} }.
\end{equation}
\end{prop}
\begin{proof}
Let the effective anticommutator conditional on $K = k$ be $\varepsilon_{k} = \frac{1}{2} \tr \big( \{A_{0}, A_{1}\} \rho_{k}^{A} \big)$. As shown in Ref.~\cite{kaniewski14} the guessing probability averaged over the two bases satisfies
\begin{equation*}
\pguess(X | K = k, \Theta) \leq \frac{1}{2} + \frac{1}{2} \sqrt{ \frac{ 1 + \abs{ \varepsilon_{k} } }{2} }.
\end{equation*}
Averaging over different values of $K$
\begin{equation*}
\pguess(X | K \Theta) = \sum_{k} p_{k} \pguess(X | K = k, \Theta) \leq \frac{1}{2} + \sum_{k} \frac{p_{k}}{2} \sqrt{ \frac{ 1 + \abs{ \varepsilon_{k} } }{2} } \leq \frac{1}{2} + \frac{1}{2} \sqrt{ \frac{ 1 + \sum_{k} p_{k} \abs{ \varepsilon_{k} } }{2} },
\end{equation*}
where we have used the concavity of the square root. For any Hermitian operator $A$ we have $\abs{ \tr (A \rho) } \leq \tr (\abs{A} \rho)$ which implies
\begin{equation*}
\sum_{k} p_{k} \abs{ \varepsilon_{k} } = \frac{1}{2} \sum_{k} p_{k} \abs[\big]{ \tr ( \{A_{0}, A_{1}\} \rho_{k}^{A} ) } \leq \frac{1}{2} \sum_{k} p_{k} \tr \big( \abs{ \{A_{0}, A_{1}\} } \rho_{k}^{A} \big) = \frac{1}{2} \tr \big( \abs{ \{A_{0}, A_{1}\} } \rho_{A} \big) = \varepsilon_{+}.
\end{equation*}
Therefore, the final bound is
\begin{equation*}
\pguess(X | K \Theta) \leq \frac{1}{2} + \frac{1}{2} \sqrt{ \frac{ 1 + \varepsilon_{+} }{2} }.
\end{equation*}
\end{proof}
\noindent It turns out that this relation is tight and can be saturated by the same setup as before, which implies that the resulting trade-off between the CHSH violation and uncertainty against classical adversaries is tight.

\subsubsection{Security definitions for WSE}
\label{sec:security-definitions}
Let $X^{n}$ be the classical register representing the $n$-bit string given to Alice and let $I$ be the classical register representing the subset of indices given to Bob.
Using the notation introduced in Section~\ref{sec:introduction} security for honest Alice means that Bob should find it difficult to guess the entire string $X^{n}$.
\begin{df}
\label{df:honest-alice}
Let $B$ be the register containing all the information that Bob might acquire during the protocol. Let $\cS_{A}$ be the set of states on registers $X^{n}, B$ that (dishonest) Bob may enforce at the end of the protocol. A WSE protocol is $(\lambda, \varepsilon)$-secure for honest Alice if the smooth min-entropy satisfies
\begin{equation*}
\hmineps(X^{n} | B) \geq \lambda n
\end{equation*}
for all $\sigma_{X^{n} B} \in \cS_{A}$.
\end{df}
\noindent Security for honest Bob, on the other hand, requires that the string $X^{n}$ takes a particular value (which Alice cannot influence anymore) and that Alice remains ignorant about the index set $\cI$ that Bob received.
\begin{df}
\label{df:honest-bob}
Let $\cS_{B}$ be the set of states on registers $X^{n}, I, A$ that (dishonest) Alice may enforce at the end of the protocol. A WSE protocol is (perfectly) secure for honest Bob if every state $\sigma_{X^{n} I A} \in \cS_{B}$ can be written as
\begin{equation*}
\sigma_{X^{n} I A} = \sigma_{X^{n} A} \otimes \frac{\mathbb{1}_{I}}{2^{n}}
\end{equation*}
for some cq-state $\sigma_{X^{n} A}$.
\end{df}
\subsection{Protocol for DI WSE and security analysis}
Since DI security can only be certified by observing some Bell violation we must make two modifications to Protocol~\ref{prot:original}: (i) we have to turn it into an entanglement-based scheme and (ii) we must introduce some way of testing the devices. The protocol we propose requires four devices in total: three for Alice and one for Bob. Below we describe the devices available to Alice.
\begin{enumerate}
\item The source emits bipartite quantum states $\rho_{AB}$. According to the ideal specification, it should emit the maximally entangled state of two qubits, i.e.~$\rho_{AB} = \ketbraq{\Phi_{+}}_{AB}$ for $\ket{\Phi_{+}}_{AB} = \frac{1}{\sqrt{2}} \big( \ket{0}_{A} \ket{0}_{B} + \ket{1}_{A} \ket{1}_{B} \big)$.
\item The main device performs one out of two binary measurements represented by observables $A_{0}, A_{1}$. According to the ideal specification, these should correspond to the computational and Hadamard basis measurements, $A_{0} = \sigma_{z}$, $A_{1} = \sigma_{x}$.
\item The test device performs one out of two binary measurements represented by observables $B_{0}, B_{1}$. According to the ideal specification, these should correspond to $B_{0} = \frac{1}{\sqrt{2}} ( \sigma_{z} + \sigma_{x} )$, $B_{1} = \frac{1}{\sqrt{2}} ( \sigma_{z} - \sigma_{x} )$.
\end{enumerate}
The only device available to Bob is a measurement device with two settings whose ideal specification coincides precisely with that of the main device of Alice (so that the outcome are identical if the measurement settings coincide).
\subsubsection{Security analysis for memoryless devices}
\label{sec:memoryless}
We call a device memoryless if it acts in the same manner every time we use it: the source always emits the same state and the measurement devices always perform the same measurements (and there are no correlations between different uses). This greatly simplifies the security analysis for several reasons: (i) we may assume that the state, measurement operators (and all quantities derived from them) are well-defined objects, (ii) probabilities can be estimated (to arbitrary precision) by repeating the experiment multiple times and (iii) testing can be completely separated from the actual protocol. In particular, the last point means that testing can be done beforehand and does need to be explicitly included in the protocol. In our protocol Alice tests her three devices by using them to violate the CHSH inequality. More specifically, she estimates the CHSH value
\begin{equation*}
\beta = \tr \big[ (A_{0} \otimes B_{0} + A_{0} \otimes B_{1} + A_{1} \otimes B_{0} - A_{1} \otimes B_{1}) \rho_{AB} \big].
\end{equation*}
We know that if $\beta \leq 2$ (no violation is observed), no security can be guaranteed and the devices cannot be used for device-independent cryptography. Therefore, from now on we assume that $\beta > 2$. While no finite set of statistical data allows Alice to determine the \emph{exact value} of $\beta$, she can estimate it to arbitrary precision which is sufficient for our analysis. Since dealing with finite statistics is not the main focus of this paper, we assume that she can actually determine $\beta$ exactly.

Recall that Proposition~\ref{prop:chsh-anticommutator} establishes a connection between the observed CHSH violation and the local incompatibility of observables (on either side). Since the test device will not take part in the actual protocol, we want to estimate the incompatibility of the main device. If $\varepsilon_{+}$ is the absolute effective anticommutator of the main device
\begin{equation*}
\varepsilon_{+} := \frac{1}{2} \tr \big( \abs{ \{A_{0}, A_{1} \} } \rho_{A} \big),
\end{equation*}
then from Proposition~\ref{prop:chsh-anticommutator} we know that
\begin{equation}
\label{eq:epsilon-plus-vs-beta}
\varepsilon_{+} \leq \frac{\beta}{4} \sqrt{ 8 - \beta^{2} }.
\end{equation}
Our goal is to show that having an upper bound on $\varepsilon_{+}$ suffices to prove security (for honest Alice) of the following DI WSE protocol.
\begin{prot}{DI WSE in the bounded/noisy storage model}
\label{prot:memoryless}
\begin{enumerate}[label=(\arabic*)]
\item Alice uses the source to generate $n$ bipartite states. She chooses a uniform $n$-bit string $\theta^{n} \in \{0, 1\}^{n}$ and uses the main device to measure the $A$ register generated in the $j$-th run with $\theta_{j}$ as the input. All the $B$ registers are passed to Bob.
\item Bob chooses a uniform $n$-bit string $\theta'^{n} \in \{0, 1\}^{n}$ and measures the $j$-th subsystem using $\theta_{j}'$ as the input to his measurement device.
\item Alice waits a fixed amount of time (this waiting time motivates the restriction on Bob's quantum memory) and then sends $\theta^{n}$ to Bob.
\item Bob determines the index set as
\begin{equation}
\label{eq:index-set}
\cI := \{ j \in [n] : \theta_{j} = \theta_{j}' \}
\end{equation}
and obtains the corresponding substring $x_{\cI}$.
\end{enumerate}
\end{prot}
It is easy to see that if the devices comply with the ideal specification, this is exactly the entanglement-based variant of Protocol~\ref{prot:original}, hence, correctness follows straightforwardly. Security argument for honest Bob is closely related to the simulation argument given in the original paper~\cite{konig12} so we just describe it informally. The correct way of defining the string $X^{n}$ is by lifting the noisy memory restriction, i.e.~we allow Bob to store all the states, wait until the receipt of the basis information and only then perform all the measurements in the correct bases. This uniquely specifies the state $\sigma_{X^{n} A}$ needed for Definition~\ref{df:honest-bob}. At the same time Bob generates a random $n$-bit string $\theta'^{n}$ and determines the index set $\cI$ through relation~\eqref{eq:index-set}. It is easy to check that this results in a uniform distribution over all possible subsets uncorrelated from the outside world (because $\theta'^{n}$ was chosen uniformly at random).

Security analysis for honest Alice turns out to be more challenging.
\begin{prop}
Protocol~\ref{prot:memoryless} executed against Bob whose quantum storage is bounded to be of dimension at most $d$ implements WSE which is $(\lambda, \varepsilon)$-secure for honest Alice for $\varepsilon = 0$ and
\begin{equation*}
\lambda \geq h( \varepsilon_{+} ) - \frac{\log d}{n},
\end{equation*}
where
\begin{equation*}
h(x) := 1 - \log \bigg( 1 + \sqrt{ \frac{ 1 + x }{2} } \, \bigg).
\end{equation*}
\end{prop}
\begin{proof}
Using the source $n$ times produces $\rho_{A^{n} B^{n}} = \bigotimes_{j = 1}^{n} \rho_{A_{j} B_{j}}$.
Alice measures all her subsystems using the main device (which produces $\rho_{X^{n} \Theta^{n} B^{n}} = \bigotimes_{j = 1}^{n} \rho_{X_{j} \Theta_{j} B_{j}}$) and then Bob measures his subsystems to obtain $K$ (which gives $P_{X^{n} \Theta^{n} K}$). It is important to emphasise that this final probability distribution is no longer of product form because Bob's measurement can introduce correlations between different rounds. First note that from Proposition~\ref{prop:equivalence} we have
\begin{equation}
\label{eq:memoryless-1}
\hmin(X^{n} | K \Theta^{n} ) \geq \hmin(X^{n} | B^{n} \Theta^{n\postinside}),
\end{equation}
where the left-hand side is evaluated on the probability distribution $P_{X^{n} \Theta^{n} K}$, while the right-hand side is evaluated on the quantum state $\rho_{X^{n} \Theta^{n} B^{n}}$. Because this quantum state is of tensor product form we have
\begin{equation}
\label{eq:memoryless-2}
\hmin(X^{n} | B^{n} \Theta^{n\postinside}) = \sum_{j = 1}^{n} \hmin(X_{j} | B_{j} \Theta_{j}\post) = n \cdot \hmin(X_{1} | B_{1} \Theta_{1}\post).
\end{equation}
where the first equality comes from the fact that the min-entropy is additive over tensor products (see Eq.~\eqref{eq:additivity}) and the second simply expresses the fact that all the rounds are identical. Now we need the bound the entropy produced while measuring a single copy of $\rho_{AB}$. Suppose that Bob measures the subsystem $B$ to produce a classical random variable $K$. From Proposition~\ref{prop:uncertainty-anticommutator} we know that the min-entropy of the probability distribution $P_{X K \Theta}$ satisfies
\begin{equation*}
\hmin(X | K \Theta) \geq h( \varepsilon_{+} ).
\end{equation*}
Since this bound is valid \emph{for all measurements that Bob might perform}, it also holds for the optimal measurement which achieves $\hmin(X | B \Theta^{*}) = \hmin(X | K \Theta)$ (see Proposition~\ref{prop:equivalence}). Therefore, we also have
\begin{equation}
\label{eq:memoryless-3}
\hmin(X | B \Theta^{*}) \geq h( \varepsilon_{+} ).
\end{equation}
Combining expressions~\eqref{eq:memoryless-1}, \eqref{eq:memoryless-2} and \eqref{eq:memoryless-3} gives
\begin{equation}
\label{eq:min-entropy-against-classical}
\hmin(X^{n} | K \Theta^{n} ) \geq n h( \varepsilon_{+} ).
\end{equation}
Finally, including the quantum memory of Bob (of dimension $d$) leads to
\begin{equation*}
\hmin(X^{n} | K Q \Theta^{n} ) \geq n h( \varepsilon_{+} ) - \log d. \qedhere
\end{equation*}
\end{proof}
\noindent Clearly, if the dimension of Bob's memory is fixed, choosing large enough $n$ brings the min-entropy rate arbitrarily close to $h(\varepsilon_{+})$.
\begin{prop}
\label{prop:noisy}
Protocol~\ref{prot:memoryless} executed against Bob whose quantum storage is represented by a quantum channel $\cF$ implements WSE which is $(\lambda, \varepsilon)$-secure for honest Alice, where $\varepsilon > 0$ is an arbitrary positive constant and
\begin{equation*}
\lambda \geq - \frac{1}{n} \log \Psucc \big( \lfloor n h( \varepsilon_{+} ) - \log (1/\varepsilon) \rfloor \big).
\end{equation*}
\end{prop}
\begin{proof}
Applying Lemma~\ref{lem:uncertainty-transmission} to Eq.~\eqref{eq:min-entropy-against-classical} (identify $X^{n} \leftrightarrow X$ and $K \Theta^{n} \leftrightarrow T$) gives
\begin{equation*}
\hmineps(X^{n} | K Q_{\textnormal{out}} \Theta^{n}) \geq - \log \Psucc \big( \lfloor n h( \varepsilon_{+} ) - \log (1/\varepsilon) \rfloor \big).
\end{equation*}
Since in the noisy storage scenario $K, Q_{\textnormal{out}}$ and $\Theta^{n}$ are the only registers available to Bob this coincides precisely with Definition~\ref{df:honest-alice}.
\end{proof}
\subsubsection{Security analysis for general devices against sequential attacks}
\label{sec:general}
As mentioned before in order to test devices that might behave differently in different rounds one must intersperse the test rounds with the live rounds. The natural solution is to introduce a biased coin-flip at the beginning of every round whose outcome determines whether the following round will be a test round (with probability $q$) or a live round (with probability $1-q$). In the previous scenario test rounds happened entirely within Alice's laboratory (using the three devices provided by Bob) and only live rounds required Alice and Bob to interaction. To make the sequential analysis conceptually simpler we give Bob even more power and allow him to operate the test box (the device used for the CHSH test), i.e.~if Alice wants to play a test round she simply sends the second input (the one she would previously use for the test device) to Bob who comes back with the outcome. Note that in this model the second part of the quantum state generated by the source always ends up with Bob (regardless of whether it is a test round or a live round), which brings us closer to the familiar scenario of two-player nonlocal games as shown in Fig.~\ref{fig:nonlocal-game}.

\begin{figure}
\begin{tikzpicture}[scale=1, line width=0.8]
	\draw[-, decorate, decoration={snake,amplitude=.8mm,segment length=4mm}] (-2.4, 0) to (1.25, 0);
	\textsq{-3}{0}{Alice};
	\draw [->] (-3, 1.2) to (-3, 0.7);
	\node[align=center] at (-3, 1.45) {$\theta$};
	\draw [->] (-3, -0.7) to (-3, -1.2);
	\node[align=center] at (-3, -1.4) {$x$};
	\textsq{1.85}{0}{Bob};
	\draw [->] (1.85, 1.2) to (1.85, 0.7);
	\node[align=center] at (1.85, 1.45) {$t/\theta'$};
	\draw [->] (1.85, -0.7) to (1.85, -1.2);
	\node[align=center] at (1.85, -1.4) {$y/k$};
\end{tikzpicture}
\caption{The key to the security proof against sequential attacks is to combine the CHSH game between the main device and the test device with the postmeasurement game between the main device and Bob's device. As already noted in~\cite{wehner08b}, if the postmeasurement game can be won perfectly, then the CHSH inequality cannot be violated. Here, we establish a complete trade-off between winning the CHSH game and the postmeasurement game. If the CHSH game can be won well, then the probability for Bob to succeed in the postmeasurement guessing game is low and hence the min-entropy about Alice's resulting string given classical information is high.}
\label{fig:nonlocal-game}
\end{figure}
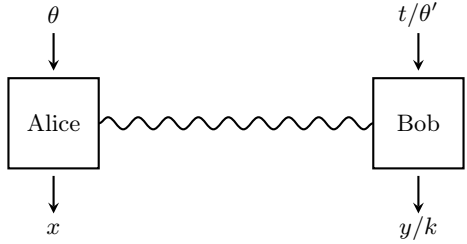
Let us stress that the interaction with the main device is always the same: regardless of whether the $j$-th round is a test round or a live round Alice always inputs a uniformly random bit $\theta_{j}$. This guarantees that the device remains ignorant whether it is currently being tested or used for a live round. On the other hand, Bob's interaction does depend on the type of round performed. Let $q_{j}$ be the bit which specifies whether the $j$-th round is a live round ($q_{j} = 0$) or a test round ($q_{j} = 1$). If Alice decides to test the devices, she will choose a random bit $t_{j}$ and request Bob to use it as an input in the CHSH game and return the outcome $y_{j}$. On the other hand, if Alice decides to play a live round, she will simply announce it to Bob and (according to the original protocol) she will not expect a response. Indeed, in the most general adversarial scenario Bob would leave his quantum system untouched and only at the end of the protocol (immediately before the memory bound) would he measure his entire system to produce some classical information $k$. Once he has received the basis information, he computes his guess as a deterministic function of $k$ and $\theta_{1}, \theta_{2}, \ldots, \theta_{n}$. In the sequential model we force Bob to produce some classical side information $k_{j}$ in every round and we require that his guess in the $j$-th round is a deterministic function (chosen before the protocol begins) of $k_{j}$, $\theta_{j}$ and any information from the previous rounds. In other words, for the $j$-th round (which we assume to be a live round) the probability of winning equals
\begin{equation*}
\Pr[ X_{j} = f_{j}( K^{j}, \Theta^{j} ) ],
\end{equation*}
where $f_{j} : ( \cK \times \{0, 1\} )^{\times j} \to \{0, 1\}$ is an arbitrary function chosen by Bob before the protocol begins. The summary of random variables generated in each round is presented in Table~\ref{table:random-variables}.
\begin{table}
\centering
\begin{tabular}{c | c | c}
every round & live round ($Q_{j} = 0$) & test round ($Q_{j} = 1$)\\
\hline
$Q_{j}, \Theta_{j}, X_{j}$ & $K_{j}$ & $T_{j}, Y_{j}$\\
\end{tabular}
\caption{The random variables generated in the $j$-th round. In every round Alice chooses the round type $Q_{j}$, generates a random input $\Theta_{j}$ and obtains an outcome $X_{j}$. If $Q_{j} = 0$ (live round) Bob generates some classical information $K_{j}$ (taking values in $\cK$). On the other hand, if $Q_{j} = 1$ (test round) Alice generates another random input $T_{j}$ and passes it to Bob who must produce an output $Y_{j}$.}
\label{table:random-variables}
\end{table}
Note that in this model the requirement of immediately producing the relevant classical information essentially replaces the need to restrict Bob's storage capabilities. The fact that success (or failure) can be assessed \emph{immediately} after every round makes such a model well-suited for a standard martingale-style analysis. It turns out that the only quantum component of such an analysis is the trade-off between the winning probabilities of the live round and the test round denoted by $p_{L}$ and $p_{T}$, respectively. Conveniently, we have already investigated this trade-off since both probabilities can be bounded through the absolute effective anticommutator $\varepsilon_{+}$. More specifically, since the probability of passing the test $p_{T}$ is related to the CHSH violation $\beta$ inequality~\eqref{eq:beta-absolute-effective-anticommutator} implies
\begin{equation}
\label{eq:pT}
p_{T}
\leq \frac{1}{2} + \frac{1}{4} \sqrt{ 1 + \sqrt{ 1 - \varepsilon_{+}^{2} } }.
\end{equation}
On the other hand, probability of winning the test round cannot exceed the optimal guessing probability of a classical adversary. Therefore, inequality~\eqref{eq:pguess-absolute-effective-anticommutator} implies
\begin{equation}
\label{eq:pL}
p_{L}
\leq \frac{1}{2} + \frac{1}{2} \sqrt{ \frac{ 1 + \varepsilon_{+} }{2} }.
\end{equation}
Combining inequalities~\eqref{eq:pT} and~\eqref{eq:pL} and treating $\varepsilon_{+}$ as a parameter taking values in $[0, 1]$ we determine the admissible pairs $(p_{L}, p_{T})$. The optimal trade-off is plotted in Fig.~\ref{fig:trade-off}.
\begin{figure}
	\centering
	\includegraphics[scale=1]{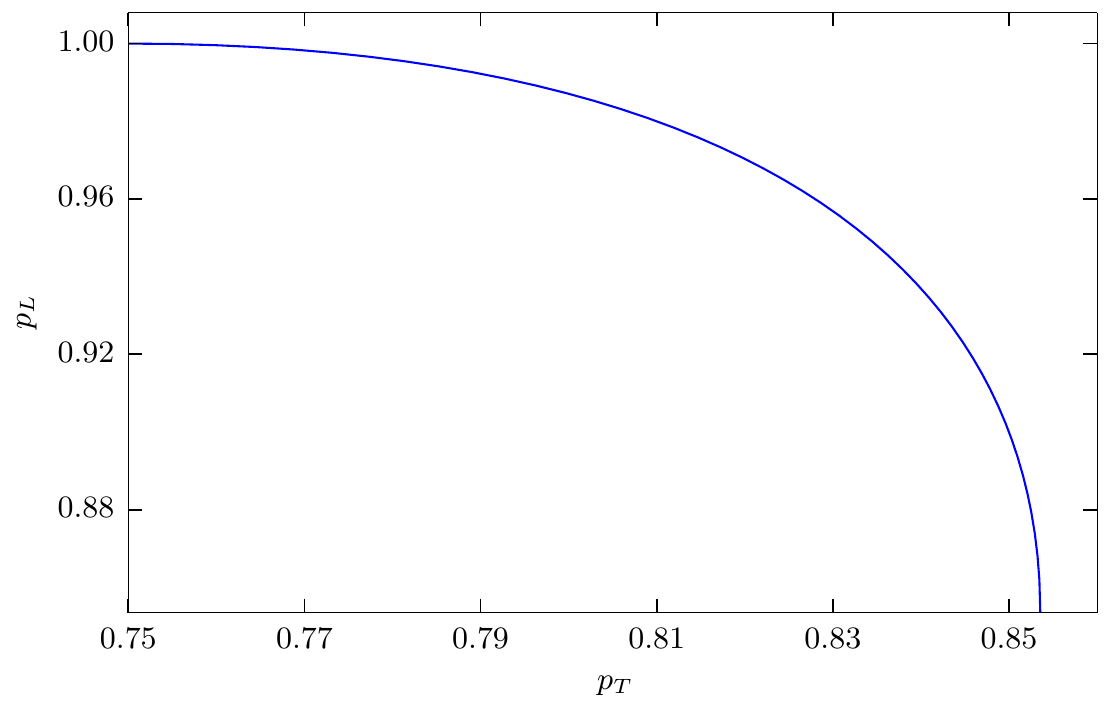}
	\caption{The (tight) trade-off between the winning probabilities of live and test rounds.}
	\label{fig:trade-off}
\end{figure}

The protocol takes three parameters: the probability of testing $q \in [0, 1]$, the CHSH threshold $\gamma \in [\frac{3}{4}, 1]$ and the number of rounds $n \in \amsbb{N}$. At the end of the protocol Alice calculates the fraction of successful CHSH rounds denoted by $\fchsh$. If $\fchsh < \gamma$ she aborts the protocol, otherwise she declares the execution correct. The security statement in this model is simply a bound on the probability that Alice believes the protocol has terminated correctly \emph{and} all the guesses of Bob are correct. We define the following random variables
\begin{gather*}
R_{l} := \sum_{j = 1}^{l} Q_{j} \nbox{(number of test rounds within the first $l$ rounds),}\\
S_{l} := \sum_{j = 1}^{l} ( X_{j} \oplus Y_{j} \oplus \Theta_{j} T_{j} \oplus 1) Q_{j} \nbox{(number of successful test rounds within the first $l$ rounds).}
\end{gather*}
Let $\cL := \{ j \in [n] : Q_{j} = 0 \}$ be the set of live rounds and for $j \in\cL$ let $G_{j}$ be the event corresponding to Bob guessing the outcome correctly, i.e.
\begin{equation}
G_{j} \iff X_{j} = f_{j}( K^{j}, \Theta^{j} ).
\end{equation}
Moreover, let $H_{l}$ be the event of guessing all the live rounds within the first $l$ rounds
\begin{equation}
H_{l} \iff \bigwedge_{j \in [l] \cap \cL} G_{j}.
\end{equation}
The failure event is defined as a conjunction of exceeding the CHSH threshold and Bob guessing all the live bits correctly
\begin{equation}
F \iff S_{n} \geq \gamma R_{n} \wedge H_{n}.
\end{equation}
Before we delve into the proof, let us show why finding an upper bound on $\Pr[F]$ is equivalent to proving security claim~\eqref{eq:security-general-2}. Let $P$ be the event of passing, i.e.~$P \iff S_{n} \geq \gamma R_{n}$ and let $p_{\textnormal{pass}} := \Pr[P]$. Writing
\begin{equation}
\Pr[ F ] = p_{\textnormal{pass}} \cdot \Pr[ H_{n} | P ]
\end{equation}
allows us to identify the last term with the sequential guessing probability conditioned on passing the test. Indeed, since
\begin{equation}
H_{n} \iff \bigwedge_{j \in \cL} X_{j} = f_{j}( K^{j}, \Theta^{j} )
\end{equation}
and assuming that Bob has chosen the optimal set of functions $\{ f_{j} \}_{j}$, we see that
\begin{equation}
\Pr[ H_{n} ] = \pguessseq( X^{\cL} | Y^{\cL} )
\end{equation}
with $Y_{j} = (K_{j}, \Theta_{j})$ being the $j$-th advice variable.

To improve clarity of the proof it is convenient to define a variable which evaluates the test threshold after $l$ rounds $X_{l} = S_{l} - \gamma R_{l}$. Note that the transition $l \to (l + 1)$ is governed by the following equation
\begin{equation}
\label{eq:1}
\begin{aligned}
\Pr[ X_{l + 1} \geq x \wedge H_{l + 1} ] &= \Pr[ X_{l} \geq x \wedge H_{l} ] \cdot (1 - q) p_{L} + \Pr[ X_{l} \geq x - (1 - \gamma) \wedge H_{l} ] \cdot q p_{T}\\
&+ \Pr[ X_{l} \geq x + \gamma \wedge H_{l} ] \cdot q (1 - p_{T}),
\end{aligned}
\end{equation}
where the three terms correspond to a successful live round, a successful test round and an unsuccessful test round, respectively. In the following proposition we establish a recursive upper bound on the probability of failure.
\begin{prop}
Let $k \geq 0$ be an arbitrary real constant. For all $l \in \amsbb{N}$ the following inequality holds
\begin{equation}
\label{eq:ansatz}
\Pr[ X_{l} \geq x \wedge H_{l} ] \leq [\alpha(q, \gamma, k)]^{l} e^{-kx},
\end{equation}
where $\alpha(q, \gamma, k)$ is a real constant defined as
\begin{equation}
\label{eq:alphak}
\alpha(q, \gamma, k) := \max_{ ( p_{L}, p_{T} ) } \big[ (1 - q) p_{L} + q e^{k (1 - \gamma)} p_{T} + q e^{- k \gamma} (1 - p_{T}) \big]
\end{equation}
and the maximisation is taken over all admissible pairs $(p_{L}, p_{T})$.
\end{prop}
\begin{proof}
Proof by induction. The statement is trivial for $l = 0$ and the second induction step follows directly from applying the ansatz \eqref{eq:ansatz} to Eq.~\eqref{eq:1} (and we directly obtain the form of $\alpha(q, \gamma, k)$ given in Eq.~\eqref{eq:alphak}).
\end{proof}
\noindent As an immediate corollary (set $x = 0$) we get a bound on the desired probability
\begin{equation}
\Pr[F] = \Pr[ S_{n} \geq \gamma R_{n} \wedge H_{n} ] = \Pr[ X_{n} \geq 0 \wedge H_{n} ] \leq [ \alpha(q, \gamma, k) ]^{n}.
\end{equation}
Since this holds for any $k \geq 0$, we choose the tightest bound
\begin{equation}
\alpha_{\textnormal{min}}(q, \gamma) := \min_{k \geq 0} \alpha(q, \gamma, k),
\end{equation}
which leads to the final bound
\begin{equation}
\Pr[ F ] \leq [\alpha_{\textnormal{min}}(q, \gamma)]^{n}.
\end{equation}
While we do not know how to find $\alpha_{\textnormal{min}}(q, \gamma)$ analytically, numerical evaluation is straightforward as explained in Appendix~\ref{app:numerical-evaluation}. Some numerical results are plotted in Fig.~\ref{fig:alphamin}.
\section{Conclusions}
We have proposed a protocol implementing DI WSE and proved security in two scenarios. In the memoryless scenario the device is first extensively tested which allows to estimate the incompatibility between the two measurements. This turns out to be sufficient to show a lower bound on the min-entropy of the output (against a classical adversary), which happens to be tight. Due to the SDP formulation of the min-entropy we can show that the lower bound is additive when multiple rounds are played (which is not obvious since Bob's attack could introduce correlations between different rounds). Moreover, we have considered a model in which the devices used by Alice might have memory but Bob is restricted to sequential attacks. In this case a martingale-style approach leads to an explicit security statement.

A secure implementation of WSE leads directly to bit commitment since the reduction involves classical post-processing only (which is trusted even in the device-independent setting). To turn WSE into some arbitrary universal functionality (e.g.~oblivious transfer) one needs to add trusted quantum communication or a secure (quantum proof) implementation of another cryptographic primitive called interactive hashing (for explicit security bounds for such constructions see Sections IV and V of Ref.~\cite{konig12}). Alternatively, one can use our techniques to directly prove security of an oblivious transfer protocol (in the bounded storage model) proposed in Ref.~\cite{damgard07}.

While this work constitutes a significant progress in the field of DI two-party cryptography, many open questions remain. In the memoryless case we only obtain bounds on the min-entropy, while it is often advantageous to derive bounds on other R\'{e}nyi entropies. The problematic step in this case is the additivity of lower bounds if multiple rounds are played. In case of the min-entropy additivity is a direct consequence of the SDP formulation (the same observation holds for the collision entropy corresponding to the pretty good measurement) but we do not know if additivity holds in general. While this problem might seem purely technical, it is of practical relevance as it would lead to significantly better security guarantees.

Another important open question is the analysis of devices with memory. In our analysis we have assumed that Bob's attack is sequential. Unfortunately, we know that sequential attacks are not always optimal (even if Alice's behaviour is sequential, see Appendix~\ref{app:sequential-guessing-not-optimal} for a simple counterexample). A security proof for devices with memory in this scenario is arguably the most important open question related to DI two-party cryptography.

Finally, we note that in the realm of the noisy-storage model there are much more sophisticated analyses~\cite{dupuis15}, which do not rely on the fact that we will first bound the adversary's information about the string $X^n$ when he is holding classical information, and subsequently relate this to his information about $X^n$ \emph{including} quantum information. Instead, one establishes a direct link between the adversary's quantum information and his uncertainty about $X^n$~\cite{dupuis15}. It remains an interesting
open question whether these techniques can be applied in the DI setting.
\section*{Acknowledgements}
We thank Rotem Arnon-Friedman, Matthias Christandl, Yfke Dulek, David Elkouss, Serge Fehr, Fabian Furrer, Iordanis Kerenidis, Le Phuc Thinh, Christian Schaffner, Florian Speelman, Marco Tomamichel, Gonzalo de la Torre and Thomas Vidick for interesting discussions. JK is supported by Ministry of Education, Singapore and the European Research Council (ERC Grant Agreement 337603). SW is funded by STW, NWO VIDI and an ERC Starting Grant.
\appendix
\section{Effective anticommutation is not sufficient against classical side information}
\label{app:effective-anticommutation-insufficient}
Let $\{ \ket{j} \}_{j = 0}^{3} \}$ be a basis for a $4$-dimensional Hilbert space. Consider the state
\begin{equation*}
\rho_{A} = \frac{1}{2} \big( \ketbraq{0} + \ketbraq{2} \big)
\end{equation*}
and binary observables
\begin{gather*}
A_{0} = \ketbraq{0} - \ketbraq{1} + \ketbraq{2} - \ketbraq{3},\\
A_{1} = \ketbraq{0} - \ketbraq{1} - \ketbraq{2} + \ketbraq{3}.
\end{gather*}
It is easy to check that the anticommutator equals
\begin{equation*}
\{ A_{0}, A_{1} \} = 2 \big( \ketbraq{0} + \ketbraq{1} - \ketbraq{2} - \ketbraq{3} \big),
\end{equation*}
which implies that the effective anticommutator equals $\frac{1}{2} \tr ( \{A_{0}, A_{1} \} \rho_{A} ) = 0$. While observable $A_{0}$ leads to no uncertainty, it is easy to verify that if we measure observable $A_{1}$ we obtain a uniform outcome. Indeed, it is possible to show non-trivial lower bound on $H_{\alpha}(X | \Theta)$.

Now, suppose that somebody holds an extra bit of classical information about the system. More specifically, we consider
\begin{equation*}
\rho_{AK} = \frac{1}{2} \big( \ketbraq{0}_{A} \otimes \ketbraq{0}_{K} + \ketbraq{2}_{A} \otimes \ketbraq{1}_{K} \big).
\end{equation*}
Measuring observable $A_{1}$ now leads to an outcome which is still uniform but it is perfectly correlated with the classical register $K$. Therefore, $H_{\alpha}(X | K \Theta) = 0$, which demonstrates that effective anticommutation does not imply uncertainty against classical side information. This is consistent with evaluating the absolute effective anticommutator $\frac{1}{2} \tr ( \abs{ \{A_{0}, A_{1} \} } \rho_{A} ) = 1$, which does not yield a non-trivial uncertainty bound.
\section{Numerical evaluation of $\alpha_{\textnormal{min}}(q, \gamma)$}
\label{app:numerical-evaluation}
Recall that our goal is to evaluate
\begin{equation*}
\alpha_{\textnormal{min}}(q, \gamma) := \min_{k \geq 0} \max_{ ( p_{L}, p_{T} ) } \big[ (1 - q) p_{L} + q e^{k (1 - \gamma)} p_{T} + q e^{- k \gamma} (1 - p_{T}) \big].
\end{equation*}
We first show that the maximisation over the admissible pairs $(p_{L}, p_{T})$ can be performed analytically. Since the expression inside the square bracket is increasing in both $p_{L}$ and $p_{T}$ the optimal point lies at the boundary, which can be parametrised by the effective absolute anticommutator $t = \frac{1}{2} \tr \big( \abs{ \{A_{0}, A_{1}\} } \, \rho_{A} \big) \in [0, 1]$. Let us restate Eqs.~\eqref{eq:pL} and \eqref{eq:pT}
\begin{align*}
p_{L}(t) &= \frac{1}{2} + \frac{1}{2 \sqrt{2}} \sqrt{ 1 + t },\\
p_{T}(t) &= \frac{1}{2} + \frac{1}{4} \sqrt{ 1 + \sqrt{ 1 - t^{2} } } = \frac{1}{2} + \frac{1}{4 \sqrt{2}} \big( \sqrt{ 1 + t } + \sqrt{ 1 - t } \big).
\end{align*}
Solving the maximisation problem corresponds to calculating $g(k) := \max_{t \in [0, 1]} f_{k}(t)$ for
\begin{equation*}
f_{k}(t) := A \sqrt{1 + t} + B \sqrt{1 - t} + C
\end{equation*}
with
\begin{align*}
A &= \frac{1}{4 \sqrt{2}} \big[ 2 (1 - q) + q e^{-k \gamma} (e^{k} - 1) \big],\\
B &= \frac{q e^{-k \gamma} (e^{k} - 1)}{4 \sqrt{2}},\\
C &= \frac{1 - q}{2} + \frac{q e^{-k \gamma} (e^{k} + 1)}{2}.
\end{align*}
The first two terms can be written as an inner product $\ave{u, v}$ for $u = (A, B)$ and $v = (\sqrt{1 + t}, \sqrt{1 - t})$. Applying the Cauchy-Schwarz inequality leads to the following upper bound
\begin{equation}
g(k) \leq \sqrt{2 (A^{2} + B^{2}) } + C,
\end{equation}
which can be achieved by choosing $t = (A^{2} - B^{2})/(A^{2} + B^{2})$.

We do not know how to minimise $g(k)$ over $k \geq 0$ analytically but numerically it is an easy task because for all valid $(q, \gamma)$ we have $g(0) = 1$ and $\lim_{k \to \infty} g(k) = \infty$ and the function is convex. Therefore, there is a unique minimum which corresponds precisely to $\alpha_{\textnormal{min}}(q, \gamma)$.

There are three cases in which we should not be able to prove security:
\begin{itemize}
\item Alice never tests: $q = 0$.
\item Alice always tests: $q = 1$.
\item The threshold is classical: $\gamma = \frac{3}{4}$.
\end{itemize}
Here, we show that in all other cases we get $\alpha_{\textnormal{min}}(q, \gamma) < 1$. Assuming that $q \neq 1$ (if $q = 1$ there is no security possible anyway) we find the Taylor expansion around $k = 0$
\begin{align*}
\sqrt{2 (A^{2} + B^{2} ) } &= \frac{1 - q}{2} + \frac{q}{4} k + O(k^{2}),\\
C &= \frac{1 + q}{2} + \frac{q ( 1 - 2 \gamma ) }{2} k + O(k^{2})
\end{align*}
and therefore
\begin{equation*}
g(k) = 1 + \Big( \frac{3}{4} - \gamma \Big) q k + O(k^{2}).
\end{equation*}
This shows that whenever $q > 0$ and $\gamma > \frac{3}{4}$ setting $k$ small enough leads to $g(k) < 1$, which concludes the argument.
\section{Sequential guessing is not necessarily optimal}
\label{app:sequential-guessing-not-optimal}
Consider a device which can be used twice and recall that we use $\Theta_{k}$ and $X_{k}$ to denote the input and output in the $k$-th round, respectively. Alice's subsystem consists of two qubits while Bob's subsystem is a qudit. The initial state is
\begin{equation*}
\rho_{A_{1} A_{2} B} = \frac{1}{2} \big( \ketbraq{0}_{A_{1}} \otimes \ketbraq{0}_{A_{2}} + \ketbraq{1}_{A_{1}} \otimes \ketbraq{1}_{A_{2}} \big) \otimes \rho_{B}
\end{equation*}
for some fixed $\rho_{B}$. Since Bob's state is uncorrelated, it carries no useful information so we can ignore it and assume that Bob picks his ``guessing functions'' deterministically. In the first round the device of Alice performs the computational basis measurement on the first qubit regardless of the value of $\Theta_{1}$. Therefore, the state after the first round is
\begin{equation*}
\rho_{X_{1} A_{2}} = \frac{1}{2} \big( \ketbraq{0}_{X_{1}} \otimes \ketbraq{0}_{A_{2}} + \ketbraq{1}_{X_{1}} \otimes \ketbraq{1}_{A_{2}} \big).
\end{equation*}
In the second round the device performs a projective measurement on the second qubit but the basis depends on $\Theta_{2}$: if $\Theta_{2} = X_{1}$ the qubit is measured in the computational basis (which ensures that $X_{1} = X_{2}$, while if $\Theta_{2} \neq X_{1}$ the qubit is measured in the Hadamard basis (which leads to$X_{1}$ and $X_{2}$ being uncorrelated). It is easy to verify that the resulting probability distribution $P_{X_{1} X_{2} \Theta_{2}}$ (we have ignored $\Theta_{1}$ since it is uncorrelated from the other random variables) is
\begin{center}
\begin{tabular}{c | c | c | c}
$\Theta_{2}$ & $X_{1}$ & $X_{2}$ & $\Pr$\\
\hline
0 & 0 & 0 & 1/4\\
0 & 1 & 0 & 1/8\\
0 & 1 & 1 & 1/8\\
1 & 0 & 0 & 1/8\\
1 & 0 & 1 & 1/8\\
1 & 1 & 1 & 1/4\\
\end{tabular}
\end{center}
In the general scenario it is optimal for Bob to guess $X_{1} = \Theta_{2}$ and $X_{2} = \Theta_{2}$ which succeeds with probability $\frac{1}{2}$.

However, in the sequential scenario Bob must attempt to guess $X_{1}$ before he learns $\Theta_{2}$. It is easy to verify that in this case his guessing probability is at most $\frac{3}{8}$.
\bibliographystyle{alphaarxiv}
\bibliography{/home/jedrek/Projekty/tex/librarysan}
\end{document}